\documentclass[11pt]{article}  
\usepackage[a4paper,margin=1in]{geometry}  
\usepackage{setspace}  
\usepackage{rotating}
\usepackage{booktabs}
\usepackage{amsmath, amssymb, amsthm, graphicx, geometry}
\usepackage{tikz}
\usetikzlibrary{positioning}
\usepackage{dsfont}
\usepackage{fix-cm}
\usepackage{enumitem}
\usepackage{bm} 
\usepackage{natbib}
\usepackage[hidelinks]{hyperref} 
\usepackage{multirow}
\usepackage{nicematrix} 
\usepackage{amsmath, amssymb, amsthm}
\usepackage{graphicx}
\usepackage{geometry}
\usepackage{authblk}
\usepackage{xr-hyper}
\usepackage{hyperref}
\usepackage{algorithm}
\usepackage[noend]{algpseudocode}
\usepackage{caption}
\usepackage{subcaption}
\usepackage{fancyhdr}
\usepackage{booktabs}
\usepackage{dsfont}
\usepackage{tikz}
\usepackage{comment}
\usepackage{cleveref}

\definecolor{bestred}{RGB}{212,0,0}
\definecolor{secondorange}{RGB}{230,120,20}
\newcommand{\bx}{\mathbf{x}}
\newcommand{\wh}{\widehat}

\newcommand{\Var}{\operatorname{Var}}
\newcommand{\Cov}{\operatorname{Cov}}

\newcommand{\E}{\mathbb{E}}
\renewcommand{\P}{\mathbb{P}}

\newcommand{\bX}{\mathbf{X}}

\newtheorem{assumption}{Assumption}
\newtheorem{theorem}{Theorem}

\newtheorem{lemma}[theorem]{Lemma}
\newtheorem{definition}{Definition}

\begin{document}

\title{\huge{Central limit theorems for the outputs of fully convolutional neural networks with time series input }}

\author{
  \textbf{Annika Betken}$^{1}$, \textbf{Giorgio Micali}$^{1}$, and \textbf{Johannes Schmidt-Hieber}$^{1}$ \\
  \texttt{a.betken@utwente.nl, g.micali@utwente.nl, a.j.schmidt-hieber@utwente.nl}
}
\date{}  
\maketitle
\footnotetext[1]{University of Twente, Faculty of Electrical Engineering, Mathematics, and Computer Science (EEMCS), Drienerlolaan 5, 7522 NB Enschede, Netherlands}


\begin{abstract}
Deep learning is widely deployed for time series learning tasks such as classification and forecasting. Despite the empirical successes, only little theory has been developed so far in the time series context. In this work, we prove that if the network inputs are generated from short-range dependent linear processes, the outputs of fully convolutional neural networks (FCNs) with global average pooling (GAP) are asymptotically Gaussian and the limit is attained if the length of the observed time series tends to infinity. The proof leverages existing tools from the theoretical time series literature. Based on our theory, we propose a generalization of the GAP layer by considering a global weighted pooling step with slowly varying, learnable coefficients.
\end{abstract}

\section{Introduction}

One way to understand the strengths and weaknesses of complex ML methods is by uncovering their inductive biases — the inherent preferences these methods exhibit when favoring certain functions over others during the learning process. Inductive bias occurs because of the choice of the underlying function class (e.g. a deep neural network with a given architecture), the (iterative) training method, and the initialization. Although these aspects are intertwined, it is meaningful to examine them individually. 

Regarding the random initialization of the network weights in deep learning, there is an extensive literature on the induced prior information, establishing Gaussian process limits of wide  neural network architectures if the input is fixed, see e.g.\ \cite{zbMATH00927558, Williams96, lee2018deep, g.2018gaussian, garriga-alonso2018deep, NEURIPS2019_5e69fda3,BasteriTrevisan2024GaussianApproxDNN}. 

A dual point of view is to keep instead the network parameters fixed and to examine the distribution of the outputs for random inputs. This distribution describes the network outputs during test time. Indeed, under the common assumption of independently generated test samples and conditionally on the trained network, the network output will follow this distribution. It is particularly insightful to understand how distributional characteristics of the inputs such as standard deviation, skewness, and kurtosis are mapped to their counterparts of the output distribution. 

Compared to the vast amount of articles on randomly initialized networks, the fixed parameter, random input setting is much less developed. A common problem considered in this area is to investigate the distribution of the network outputs for (mixtures of) Gaussians as input,  see \cite{abdelaziz:hal-01162550, 8579046, ZHANG2021170}. More general input distributions were only recently discussed in 
\cite{Krapf_Hagn_Miethaner_Schiller_Luttner_Heinrich_2024, kofnov2025exact}. All this work focuses on an explicit characterization of the output distribution. A drawback are the involved expressions that arise for more complex architectures.

Both the distribution of stationary time series and  
convolutional layers exhibit translation invariance. For this reason, network architectures based on 
convolutional layers have become a standard tool in time series classification
and related tasks; see, for instance, \cite{2016arXiv160306995C, TANG2021113544}. Particularly successful are fully convolutional networks (FCNs) with global average pooling layer (GAP), \cite{7966039, 10.1145/3649448, middlehurst2024}.

Working in the fixed parameter, random input setting, and considering short-range dependent linear processes as inputs, we establish in this paper the Gaussian limit distribution of FCNs with GAP layer. The limit is with respect to the length of the input time series and does not require infinite network width. To prove the result, we leverage techniques from the time series literature. Let $\mathbf{X}_1,\mathbf{X}_2,\ldots \in \mathbb{R}^k$ be the vector-valued input time series. Rather than working directly with a specific CNN architecture, we more generally establish asymptotic normality for the weighted sum process
\begin{equation}
\label{eq:objective}
S_n
:=
\sum_{t=1}^n \sum_{m=1}^M
a_{t,m}^{(n)}
\Big(
f_m(\mathbf X_t) - \mathbb E[f_m(\mathbf X_t)]
\Big),
\end{equation}
with deterministic weights $(a_{t,m}^{(n)})$ and deterministic Lipschitz continuous functions
$f_m:\mathbb R^k\to\mathbb R.$ Our main theoretical result establishes a central limit theorem for $S_n$ for Lipschitz continuous functions $f_m$ and under suitable
conditions on the dependence structure of $\mathbf X_1, \mathbf X_2, \ldots$

More generally, we show that the outputs of FCNs with GAP layers are still multivariate Gaussian if we do allow for slowly varying weights. We therefore propose learnable and slowly varying weights in the GAP layer as a new method, where the amount of variation of the weights is controlled via a first-order difference penalty. Studying the performance of this method across a variety of benchmark datasets shows that the new method leads to improved classifiers in about 30\% of the considered datasets.

\subsection*{Notation }
Vectors are denoted by bold letters. For a vector,  $\|\cdot\|$ denotes the Euclidean norm and for a matrix $A$, 
\(
\|\cdot \|
\) denotes the spectral norm.  Weak convergence is denoted by $\xrightarrow{\mathcal{D}}$. For two sequences $(a_n)_{n\geq 1}$ and $(b_n)_{n\geq 1}$, we use the standard notation
$a_n = o(b_n) $
to express that 
$
\frac{a_n}{b_n} \longrightarrow 0 $
as $n\to\infty$. A function $f : \mathbb{R}^k \to \mathbb{R}$ is \emph{Lipschitz} with Lipschitz constant $\text{Lip}(f)$ if
\[
\mathrm{Lip}(f) := \sup_{\mathbf{x} \neq \mathbf{y}} 
\frac{|f(\mathbf{x}) - f(\mathbf{y})|}{\|\mathbf{x} - \mathbf{y}\|} < \infty.
\]

\section{Asymptotic normality for weighted sums of linear processes}

To model the network input, we consider a
class of stationary processes that is sufficiently rich to model temporal
dependence, yet tractable from an asymptotic point of view. In particular, the processes can be non-Gaussian.

\begin{assumption}
\label{assum1:M1}
Assume that $(\mathbf{X}_t)_{t\geq 1}$ is an $\mathbb{R}^k$-valued linear process of the form $$\mathbf{X}_t = \sum_{j=0}^\infty A_j\mathbf{Z}_{t-j}$$  
where $(\mathbf Z_t)_{t\in\mathbb Z}$ is an i.i.d.\ sequence of centered random variables such that $\mathbb{E}[\mathbf{Z}_0] = 0,$ $\mathbb{E}[\|\mathbf{Z}_0\|^2] < \infty$, and $k\times k$ matrices $A_j$ satisfying \(
\sum_{j=0}^\infty \|A_j\|< \infty.
\)
\end{assumption}

Under these conditions, $(\mathbf{X}_t)_{t\geq 1}$ is well defined, short range dependent, and
strictly stationary, see for instance Chapter 1 in \cite{BrockwellDavis1991}.

\begin{theorem}\label{thm:main}
Let $(\mathbf X_t)_{t\ge 1}$ be a linear process satisfying
Assumption~\ref{assum1:M1}, and let
$f_1,\ldots,f_M:\mathbb{R}^k\to\mathbb{R}$ be Lipschitz continuous functions.
Let $\big(a_{t,m}^{(n)}\big)_{t=1,\dots,n;\,m=1,\dots,M}$ be deterministic weights such that
$a_{t,\cdot}^{(n)}=0$ whenever $t\notin \{1,\ldots,n\}.$ Assume that for any $m,\ell\in\{1,\ldots,M\}$ and any $h\in\mathbb{Z}$,
\begin{equation}
\label{eq:W1-multi}
G_{n,m\ell}(h)
:=
\frac{1}{n}\sum_{t=1}^n a_{t,m}^{(n)} a_{t+h,\ell}^{(n)}
\;\xrightarrow[n\to\infty]{}\;
G_{m\ell}(h)\in\mathbb{R}
\tag{W1}
\end{equation}
and
\begin{equation}
\label{eq:W2-multi}
\max_{1\le t\le n}\max_{1\le m\le M} \big|a_{t,m}^{(n)}\big|
=o(\sqrt{n})
\qquad\text{as }n\to\infty.
\tag{W2}
\end{equation}
If\[
S_{m,n}
:=
\sum_{t=1}^n 
a^{(n)}_{t,m}\Big(f_m(\mathbf X_t)-\mathbb{E}[f_m(\mathbf X_t)]\Big),
\]
then,
\[
\frac{1}{\sqrt{n}}\,\begin{pmatrix}
S_{1,n}\\
\vdots\\
S_{M,n}
\end{pmatrix}
\;\xrightarrow{\mathcal D}\;
\mathcal N(0,\Sigma),
\]
where $\Sigma=(\Sigma_{m\ell})_{m,\ell=1}^M$ is the covariance matrix with entries
\[
\Sigma_{m\ell}
=
G_{m\ell}(0)\,\Cov\!\big(f_m(\mathbf X_1),f_\ell(\mathbf X_1)\big)
+
2\sum_{h=1}^\infty
G_{m\ell}(h)\,\Cov\!\big(f_m(\mathbf X_1),f_\ell(\mathbf X_{1+h})\big)<\infty\;.
\]
\end{theorem}

Without non-linearity $f$, weighted sums of linear processes have been studied in \cite{abadir2014asymptotic}. Using that $\mathbf X_t=\sum_{j=0}^ \infty A_j \mathbf Z_{t-j}=\sum_{q=-\infty}^t A_{t-q}\mathbf Z_q,$ one can then rewrite 
\begin{align*}
    \sum_{t=1}^n a_{t,m}^{(n)} \mathbf X_t
    = \sum_{q=-\infty}^n\Big(\sum_{t=1}^n a_{t,m}^{(n)}A_{t-q}\mathbf{1}(t\geq q)\Big)\mathbf{Z}_q,
\end{align*}
simplifying the analysis considerably.

Theorem 2.1 of \cite{Furmanczyk} corresponds to Theorem \ref{thm:main} with $M=1$ and weights $a_t^{(n)}= b_n \mathbf{1}(t\leq n)$ for $(b_n)_{n\in \mathbb{N}}$ a real sequence satisfying $b_n \to 1$. To see this, observe that in this scenario the conditions of Theorem \ref{thm:main} are satisfied since for any $h=0,1,2, \ldots $, we have   
    $\tfrac{1}{n}\sum_{t=1}^{n-h}b_n\,b_n=  b_n^2 \tfrac{n-h}{n}\to 1\;,$
    and $\max\limits_{1\le t\le n} |a_t^{(n)}| = o(\sqrt{n})$. Thus, as $n\to \infty$, 
    \[
\frac{b_n}{\sqrt{n}}\sum_{t=1}^n \big(f(\mathbf{X}_t)-\E[f(\mathbf{X}_t)]\big)\xrightarrow{\mathcal{D}}\mathcal{N}\Big(0, \Var\big(f(\mathbf{X}_1)\big)+2 \sum_{w=1}^\infty \,\Cov\big(f(\mathbf{X}_1),f(\mathbf{X}_{1+w})\big)\Big).
\]

\section{Fully convolutional networks (FCNs)}
\label{sec:fcns}

Convolutional layers underlie various state-of-the-art network architectures. As pointed out in Section 11.2 of \cite{Bengio2013}, they are particularly well suited for time series data, image data, audio spectrograms, and video data. Indeed, stationarity implies translation invariance of the finite-dimensional
distributions, which makes convolutional architectures a natural choice when the
goal is to learn characteristics of the distribution.

Here we consider a $\mathbb{R}^d$-valued sequence $\mathbf x=(\mathbf x_1,\ldots,\mathbf x_n)$ as input, that is, $\mathbf x_j\in\mathbb R^d$ for all $j=1,\ldots,n$. As is common in the deep learning literature, we apply zero-padding, that is, we regard $\mathbf x$ as an infinite sequence by defining $\mathbf x_j:=0$ for $j>n$. Let $W_0,\ldots,W_{k-1}$ be $k$ learnable convolutional filters of size $m\times d$, let $\mathbf b\in\mathbb R^m$ be a learnable bias vector, and let $\sigma:\mathbb R\to\mathbb R$ be a given activation function such as the ReLU activation function $\sigma(x)=\max\{x,0\}$. The corresponding convolutional layer transforms the input sequence $\mathbf x=(\mathbf x_1,\ldots,\mathbf x_n)$ into the $\mathbb R^m$-valued time series
\begin{align}
    \mathbf Z_t:=\sigma\!\Big(\mathbf b+\sum_{j=0}^{k-1}W_j\mathbf x_{t+j}\Big),\qquad t=1,\ldots,n,
\end{align}
where the activation function is applied component-wise. The convolution smooths the time series: Neighboring output values tend to be more similar than the input values.
A fully convolutional neural network (FCN) stacks convolutional layers with potentially different width on top of each other.

Residual connections were introduced to address fundamental optimization
difficulties arising in very deep architectures, \cite{He_2016_CVPR}. Residual connections allow
layers to learn perturbations of the identity map enabling stable training of substantially deeper
networks. Residual connections have also been successfully adapted for time - series
analysis, \cite{7966039}. In this context, residual connections can be interpreted as allowing the network
to preserve low-level temporal information while progressively refining it through
nonlinear transformations.

\begin{figure}[h] \centering \includegraphics[width=1\linewidth]{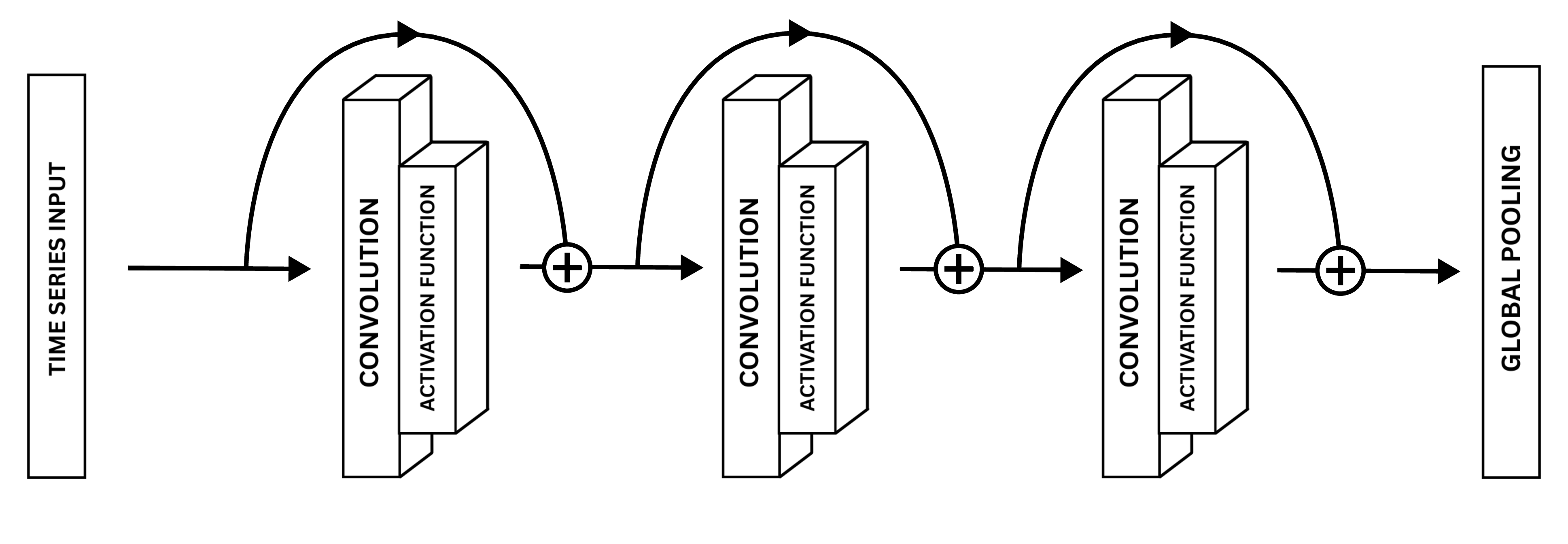} \caption{FCN architecture with Residual connection.} \label{fig:ResNet} \end{figure}

A representation of such architecture is displayed in Figure \ref{fig:ResNet}, cf. also Figure 1 in \cite{7966039}. 

\begin{definition}[FCNs with residual connections, $L$ layers, filter widths $k_1,\ldots, k_L,$ and layer widths $m_1,\ldots, m_L$] 
  Let $\mathbf{x}=(\mathbf{x}_1,\ldots,\mathbf{x}_n)\in \mathbb{R}^{d\times n}$ be the input and let $(\mathbf H^{(0)}_1, \mathbf H^{(0)}_2, \ldots):= \newline (\mathbf{x}_1,\ldots,\mathbf{x}_n, 0, 0, \ldots)$ be the input extended by zero padding. For any layer $\ell=1,\ldots,L$, and any $j=0,\ldots,k_\ell-1$, let $\mathbf b^{(\ell)}\in\mathbb{R}^{m_\ell}$ be a bias vector
and let $I^{(\ell)}, W^{(\ell)}_j\in\mathbb{R}^{m_\ell\times m_{\ell-1}}$ be matrices ($m_0:=d$). Define recursively the convolutional map
\begin{equation}\label{eq:def_conv_block}
\mathbf H^{(\ell)}_t
:= I^{(\ell)}\mathbf{H}^{(\ell-1)}_t +
\sigma\!\Big(
\mathbf b^{(\ell)} + \sum_{j=0}^{k_\ell-1} W^{(\ell)}_j\,\mathbf H^{(\ell-1)}_{t+j}
\Big)
\in\mathbb{R}^{m_\ell},
\qquad t=1,2,\ldots, \ \  \ell=1,\ldots,L,
\end{equation}
where $\sigma:\mathbb{R}\to\mathbb{R}$ is applied to vectors component-wise. The output of the FCN is the multivariate time series $\mathbf H^{(L)}_1,\ldots,\mathbf H^{(L)}_n\in \mathbb{R}^{m_L}.$
\end{definition}
 To reduce the final network output to one $m_L$-dimensional vector, FCNs commonly deploy global average pooling (GAP) as output layer. In a FCN with GAP layer, the network output is
\begin{equation}
\mathrm{GAP}(\mathbf{x})
:=
\frac{1}{n}\sum_{t=1}^{n}\mathbf H^{(L)}_t
\in \mathbb{R}^{m_L}.
 \label{eq.GAP}
\end{equation}

The matrices $I^{(\ell)}$ model the residual connections and are not learned. Setting them to zero corresponds to the case that no residual connections are included. In case that $m_\ell=m_{\ell-1},$ a residual connection corresponds to taking $I^{(\ell)}$ as the identity matrix. Including residual connections in the case $m_\ell\neq m_{\ell-1}$ requires to choose the residual connection matrix as a generalization of the identity matrix to non-square matrices.

The entries of the biases $\mathbf{b}^{(1)},\ldots, \mathbf{b}^{(L)}$ and the weight matrices $W_j^{(\ell)}$ are the learnable parameters and are jointly trained
via gradient descent. For time-series classification, the pooled feature vector produced by the GAP layer
\eqref{eq.GAP} has dimension $m_L$ and is fed into a final affine 
layer with $K$ outputs (matching the number of classes) via
\begin{equation}
\label{eq.dense_layer}
    \mathbf y
=
A\,\mathrm{GAP}(\mathbf{x})+\mathbf b\in\mathbb{R}^K,
\end{equation}
where $A\in\mathbb{R}^{K\times m_L}$ and $\mathbf b\in\mathbb{R}^K$ are learnable
parameters. The class probabilities are then obtained by applying the softmax 
\begin{align}
\label{eq.soft_max}
    (y_1,\ldots,y_{K})
    \mapsto \Big(\frac{e^{y_1}}{e^{y_1}+\ldots+e^{y_{K}}},
    \ldots, \frac{e^{y_{K}}}{e^{y_1}+\ldots+e^{y_{K}}}\Big)
\end{align}
transforming the output into a $K$-dimensional probability vector.

It is also common to include batch normalization in every convolutional layer. For batch normalization, the means and standard deviations of the activations in a layer are estimated based on a batch of the training sample. If $\wh \mu$ and $\wh \sigma$ are the respective estimates for one neuron, batch normalization transforms the activation using the map $z\mapsto (z-\wh \mu)/\sqrt{\wh \sigma^2+\epsilon}$ for a small fixed $\epsilon$ that guarantees stability for small standard deviations.

A systematic review of deep learning approaches to time series classification can be found in \cite{Ismail2019}. GAP has been introduced in Section 3.2 of \cite{2013arXiv1312.4400L} and used for time series classification in \cite{7966039}. GAP received acclaim in the highly cited paper \cite{Zhou_2016_CVPR} and the recent comparison in \cite{meyer2025deep} shows that GAP remains competitive on benchmark datasets. 

Common architectures consist of few FCN layers with kernel lengths below $10$ and filter counts in the hundreds, see e.g.\ Section 3.2.2 of \citet{Ismail2019}.

\begin{lemma}\label{lem:fcn_res_representation}
Consider an FCN defined in \eqref{eq:def_conv_block}, with Lipschitz continuous activation function \(\sigma\), \(L\) layers, filter widths \(k_1,\ldots,k_L\), and layer widths \(m_1,\ldots,m_L\).
Set
\begin{equation}
\label{eq:k_ell}
K_\ell:=\sum_{q=1}^{\ell} k_q-(\ell-1),\qquad \ell=1,\ldots,L.
\end{equation}
Then, for every \(\ell=1,\ldots,L\), there exists a deterministic Lipschitz continuous function
\[
\mathbf g_\ell:\mathbb{R}^{K_\ell}\to\mathbb{R}^{m_\ell},
\]
that is independent of \(t\), such that for all \(t\ge 1\),
\begin{equation}\label{eq:induction_claim_fcn}
\mathbf H_t^{(\ell)}
=
\mathbf g_\ell\bigl(x_t,\ldots,x_{t+K_\ell-1}\bigr).
\end{equation}
\end{lemma}

The proof can be found in Appendix \ref{appendix:missing_proofs}.
As a result of Lemma \ref{lem:fcn_res_representation}, the activations in the last layer $\mathbf H^{(L)}_t$ only depend on the finite input window
$\mathbf x_{t:K_L}:=(x_t,\ldots,x_{t+K_L-1})$, and
\[
\mathrm{GAP}(\bx)= \frac{1}{n}\sum_{t=1}^{n} \mathbf{H}_t^{(L)}
=
\frac{1}{n}\sum_{t=1}^{n}\mathbf g_L\big(\mathbf x_{t:K_L}\big).
\]
 Based on this representation, one can now apply Theorem  \ref{thm:main}.
 
\begin{theorem}\label{thm:fcn_gap}
Consider an FCN architecture as in the previous lemma and write $\mathbf g_L=(g_1,\ldots,g_{m_L})^\top:\mathbb{R}^{K_L}\to\mathbb{R}^{m_L}.$ Assume that the input is $\bX:=(\bX_1,\ldots,\bX_n)$ with $(\mathbf X_t)_{t\ge 1}$ a linear process satisfying
Assumption~\ref{assum1:M1}. Then, as $n \to \infty,$
\[
\sqrt n\Big(\mathrm{GAP}(\bX)-\E\big[\mathrm{GAP}(\bX)\big]\Big)
\ \xrightarrow{\ \mathcal D\ }\ 
\mathcal N \!\big(\mathbf 0,\Sigma^{\mathrm{GAP}}\big),
\]
where the $(j,\ell)$-th entry of the $m_L \times m_L$ covariance matrix $\Sigma^{\mathrm{GAP}}$ is given by
\begin{equation}
\label{eq:long_run_cov}
    \Sigma^{\mathrm{GAP}}_{j,\ell}
=
\Cov\!\big(g_j(\mathbf X_1),g_\ell(\mathbf X_1)\big)
+
2\sum_{h=1}^\infty
\Cov\!\big(g_j(\mathbf X_1),g_\ell(\mathbf X_{1+h})\big).
\end{equation}

\end{theorem}

The proof can be found in Appendix \ref{appendix:missing_proofs}.

The statement proves that conditionally on a trained FCN, the output distribution of the GAP layer during inference time is asymptotically normal for long time series inputs. The theorem provides moreover a formula to compute the covariance matrix. For classification one passes this through the affine transformation \eqref{eq.dense_layer} and a softmax layer, making the asymptotic distribution to be the logistic normal,  \cite{zbMATH03675128}.

In simple cases, one can derive closed-form expressions for the function $\mathbf g.$ As example, consider a FCN
consisting of two successive one-dimensional convolutional layers with filter
lengths $m$ and $k$, weight  vectors $\mathbf w=(w_1,\ldots,w_m)^\top \in\mathbb{R}^{m}$ and
$\widetilde{\mathbf w}=(\widetilde w_1,\ldots, \widetilde w_{k})^\top\in\mathbb{R}^{k}$, biases $v, \widetilde v,$ and activation function $\sigma$. In this case, $a(j)=j, b(j)=j+m+k-2,$ $N=n,$ and
\begin{equation}
\label{eq:f_two_layer_cnn}
g\big(u_1,\ldots,u_{m+k-1}\big)
=
\sigma\!\Bigg(
\sum_{\ell=1}^{k}
\widetilde{w}_\ell\,
\sigma\!\Big(\sum_{r=1}^{m} w_r u_{\ell+r-1}\Big)
\Bigg).
\end{equation}

\section{The GAP layer with varying weights}\label{sec:weighted_gap}

We propose to replace the GAP layer $\mathrm{GAP}(\mathbf x)
:=\tfrac{1}{n}\sum_{t=1}^n \mathbf H^{(L)}_t$ by a weighted pooling layer
\begin{equation}\label{eq:wgap_def}
\mathrm{WGAP}(\mathbf x)
:=\sum_{t=1}^n a_{t,n}\,\mathbf H^{(L)}_t
\in\mathbb{R}^{m_L},
\end{equation}
where $\mathbf a_n:=(a_{1,n},\ldots,a_{n,n})^\top\in\mathbb{R}^n$are learnable parameters. The choice $a_{t,n}\equiv 1/n$ recovers standard GAP.

Theorem~\ref{thm:main} implies that the asymptotic Gaussianity of the pooled output is preserved
for a broad class of deterministic weight sequences. In particular, Gaussian limits still hold provided that all weights are $o(\sqrt{n})$ and that $\tfrac{1}{n}\sum_{t=1}^n a_{t,n}a_{t+h,n}$ converges for any $h=0,1,\ldots$ as $n \to \infty.$ A sufficient condition guaranteeing this convergence is that the successive weights vary slowly. Thus, we propose to incorporate the coefficients $\mathbf{a}_n$ as learnable parameters in the network and to add the first-order difference penalty
\begin{align}
    \lambda \sum_{j=1}^{n-1} \big(a_{(j+1),n}-a_{j,n}\big)^2
    \label{eq.penalty}
\end{align}
as regularization term to the training loss, where $\lambda>0$ is a regularization parameter. Zero penalty is given to solutions $a_{1n}=\ldots = a_{nn}.$ Large $\lambda$ favors learning network configurations with similar subsequent weights. This means that the sequence of learned weights $(a_{1,n},\ldots, a_{n,n})$ has small quadratic variation. 

As common in ML, the parameter $\lambda$ can be chosen via a validation dataset or by cross-validation.

Dynamic temporal pooling \cite{Lee_Lee_Yu_2021} is a related method. It, however, splits the sequence $1, \ldots, n$ into blocks and then averages on each block.

\section{Numerical experiments}
\label{sec:simulations}
\subsection{Numerical simulations of the FCN output distribution}

We illustrate the asymptotic normality of the FCN networks with GAP layer. The network consists of $B$ sequential residual convolutional blocks. We consider ReLU networks and sigmoid networks.
Each block is composed of two one-dimensional convolutions with filter lengths
$3$ and $2$, respectively, 
and an identity skip connection.
All convolutional weights are initialized using He initialization and kept
fixed as no training is performed. The randomness in the output distribution stems solely
from the input time series $(X_t)_{t\geq1}$, which is generated either according to a one-dimensional AR($1$) process $X_t = 0.6 X_{t-1} + \varepsilon_t$ with i.i.d.\ standard Gaussian innovations $\varepsilon_t \sim \mathcal N(0,1),$ or according to a MA$(1)$ process $X_t = 0.6\epsilon_{t-1} + \varepsilon_t$ with i.i.d.\ $\varepsilon_t \sim \mathcal N(0,1).$

Both processes satisfy Assumption \ref{assum1:M1}, see Lemma 1 in \cite{betken2025ordinal}. 

For both settings,  Figure~\ref{fig:simulation_histogram} shows the close agreement of the network outputs and the Gaussian distribution, corroborating the derived asymptotic normality result.

\subsection{Numerical estimation of the covariance structure of GAP outputs}

The covariance matrix $\Sigma^{\mathrm{GAP}}$ in Theorem~\ref{thm:fcn_gap} captures both the  variability of the GAP output channels and their cross-dependencies. To estimate $\Sigma^{\mathrm{GAP}}$, we use the kernel-based long-run covariance estimator proposed by \cite{dejong2000consistency}.  
For $h\ge 0$ and $K_L$ as defined in \eqref{eq:k_ell}, define the sample mean and the sample lag-$h$ covariance matrix by
\begin{align*}
    &\bar{\mathbf g}_L:=\frac{1}{n-K_L+1}\sum_{s=1}^{n-K_L+1}\mathbf g_L(X_s,\ldots,X_{s+K_L-1})\\
    &\widehat\Gamma_n(h):=
\frac{1}{n-K_L+1}\sum_{t=1}^{n-K_L+1-h}
\big(
\mathbf g_L(X_t,\ldots,X_{t+K_L-1})
-\bar{\mathbf g}_L
\big)
\big(
\mathbf g_L(X_{t+h},\ldots,X_{t+h+K_L-1})
-\bar{\mathbf g}_L\big)^\top,
\end{align*}
and for $h>0$ let $\widehat\Gamma_n(-h):=\widehat\Gamma_n(h)^\top$. Given a kernel $k:\mathbb R\to\mathbb R$ and a bandwidth sequence $(b_n)_n$ such that $b_n\to\infty$, define
\[
\widehat\Sigma_n^{\mathrm{GAP}}
:=
\widehat\Gamma_n(0)
+
\sum_{h=1}^{n-K_L}
k\!\left(\frac{h}{b_n}\right)
\bigl(
\widehat\Gamma_n(h)+\widehat\Gamma_n(h)^\top
\bigr).
\]
Under standard regularity conditions on the kernel, the bandwidth, and the process $(\mathbf X_{t:K_L})_{t\geq 0}$, \cite{dejong2000consistency} proved that $\widehat\Sigma_n^{\mathrm{GAP}}$ is a consistent estimator for the covariance matrix $\Sigma^{\mathrm{GAP}}$. 

For all numerical experiments, we use the Bartlett kernel
\(
k(x)=(1-|x|)\mathbf 1_{\{|x|\le 1\}},
\)
which belongs to the kernel class considered in their Assumption~1. 

Consider as input time series \((X_t)_{t}\) a Gaussian AR\((1)\) process
\(
X_t=\theta X_{t-1}+\varepsilon_t,\) \(\theta\in[0,1),
\) with independent innovations \(\varepsilon_t\sim\mathcal N(0,\sigma_\varepsilon^2)\). For $m_1$-dimensional vector \(W_0^{(1)}\), a one-layer convolutional network with $k_1=1$ is given by
\begin{equation}
\label{eq:toymodel}
\mathrm{GAP}(\mathbf x)
=
\frac1n\sum_{t=1}^n \mathrm{ReLU}\!\left(W_0^{(1)}X_t\right).
\end{equation}
Lemma~\ref{lem:cnn_trivial_case} in the Appendix derives expressions for the entries of the long-run correlation matrix
\[
\mathcal R_{j,\ell}^{\mathrm{GAP}}(\theta)
=
\frac{\Sigma_{j,\ell}^{\mathrm{GAP}}(\theta)}
{\sqrt{\Sigma_{j,j}^{\mathrm{GAP}}(\theta)\Sigma_{\ell,\ell}^{\mathrm{GAP}}(\theta)}}, \quad j,\ell=1,\ldots,m_1,
\]
covering arbitrary $\theta\in [0,1)$ and filter width $k_1\geq 1.$ For the boundary values $\theta=0$ and $\theta \uparrow 1,$ we specifically obtain
\[
\mathcal R_{j,\ell}^{\mathrm{GAP}}(0)
=
\begin{cases}
1, & \text{if } W^{(1)}_{0,j}W^{(1)}_{0,\ell}\geq 0,\\[1.2ex]
-\dfrac{1}{\pi-1}\approx -0.47, & \text{if } W^{(1)}_{0,j}W^{(1)}_{0,\ell}<0,
\end{cases}
\]
and
\[
\lim_{\theta\uparrow 1}\mathcal R_{j,\ell}^{\mathrm{GAP}}(\theta)
=
\begin{cases}
1, & \text{if } W^{(1)}_{0,j}W^{(1)}_{0,\ell}\geq0,\\[1.2ex]
\displaystyle \frac{\log 2-2}{\pi+\log 2-2} \approx - 0.71, & \text{if } W^{(1)}_{0,j}W^{(1)}_{0,\ell}<0.
\end{cases}
\]
This reveals that in these cases a large fractions of the network outputs are perfectly correlated, see also Figure~\ref{fig:gap_corr_theta_comparison} for the numerical results.

\begin{figure}[t]
    \centering
    \begin{minipage}[t]{0.35\linewidth}
        \centering
  
        \vspace{0.3em}
        \includegraphics[width=\linewidth]{ 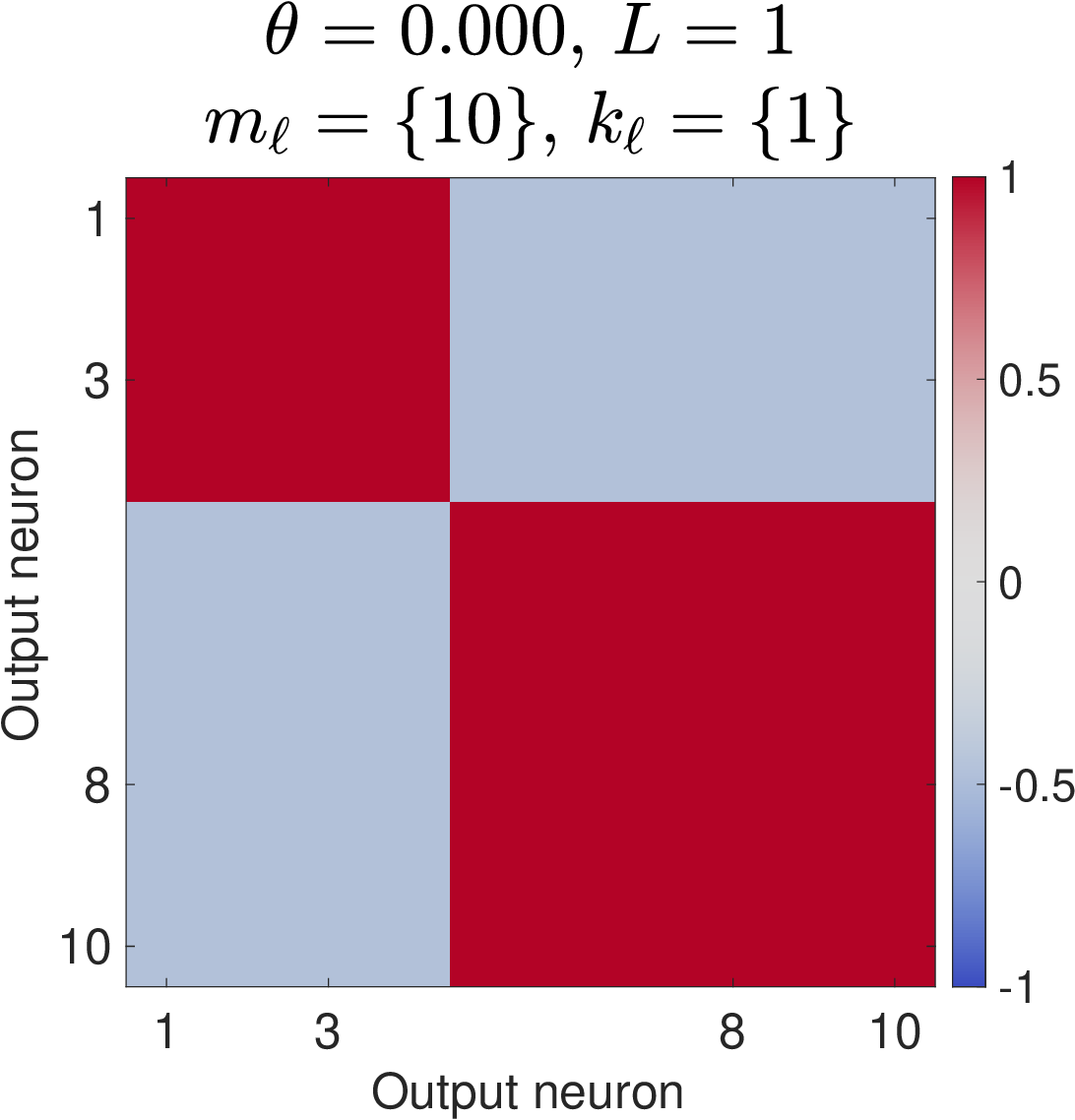}
    \end{minipage}
    \begin{minipage}[t]{0.34\linewidth}
        \centering
        
        \vspace{0.3em}
        \includegraphics[width=\linewidth]{ 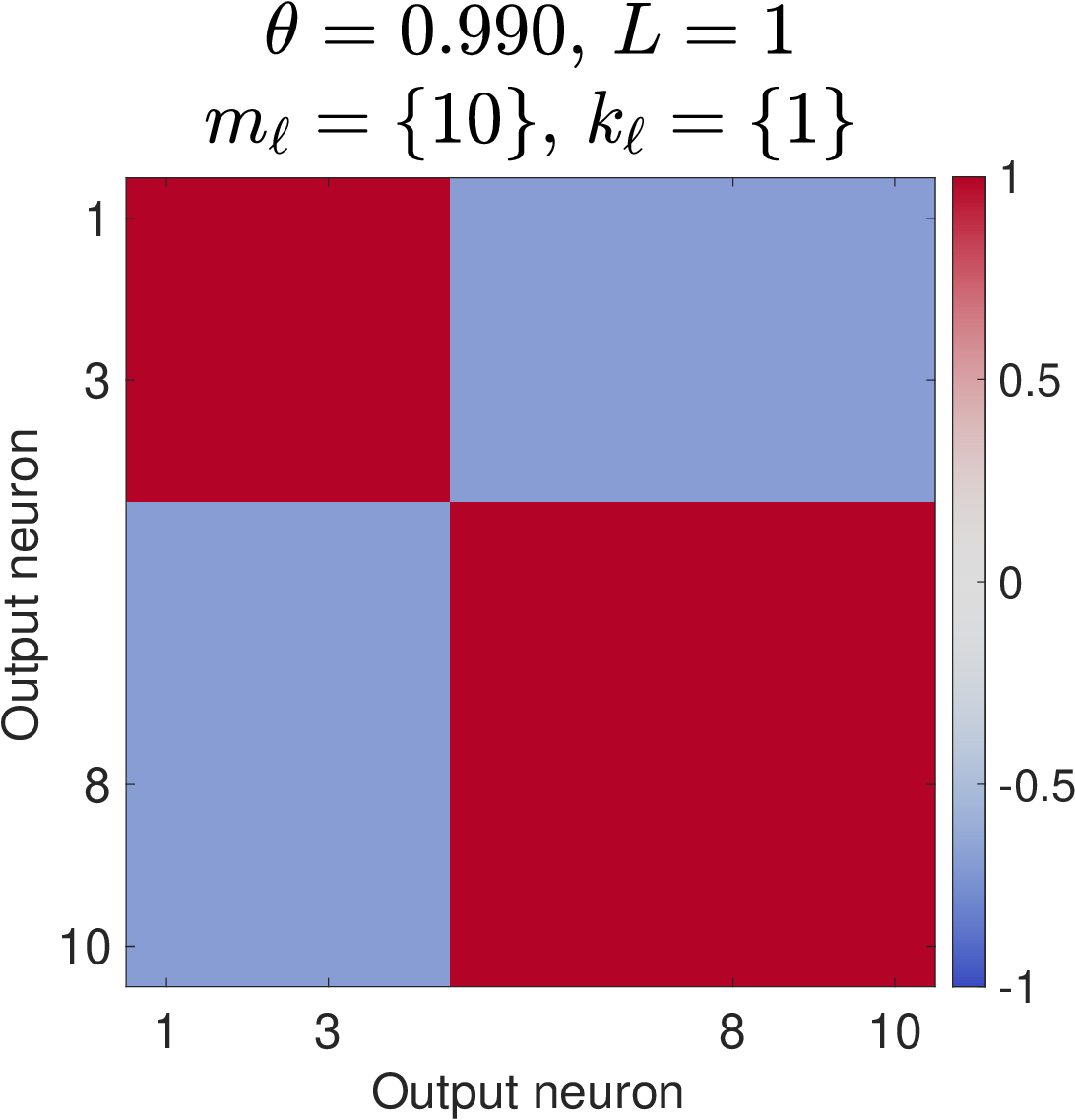}
    \end{minipage}
    \caption{Empirical autocorrelation matrices \(\widehat{\mathcal R}^{\mathrm{GAP}}(\theta)\) for the one-layer GAP network with \(m_1=10\) under fixed Gaussian initialization and input generated from an AR\((1)\) process with parameters \(\theta=0\)  (left panel) and  \(\theta=0.99\) (right panel). The numbering of the output neurons has been chosen to maximize the correlation of successive neurons, thereby revealing the strong clustering of the outputs. }
    \label{fig:gap_corr_theta_comparison}
\end{figure}
For Gaussian AR(1) input with $\sigma_\varepsilon=5$ and dependence parameter \(\theta=0.65,\) Figure~\ref{fig:settingB_all} studies how the correlation structure changes as the architecture becomes deeper ($L\in \{1,2,3,4\}$). All cases are for ReLU activation, without residual connections, and for He initialization that is kept fixed throughout the experiment. With growing depth and width, the block structure of the correlation matrix successively disappears.

\begin{figure}[!t]
    \centering

    \begin{subfigure}[t]{0.4\textwidth}
        \centering
        \includegraphics[height=5.2cm]{ 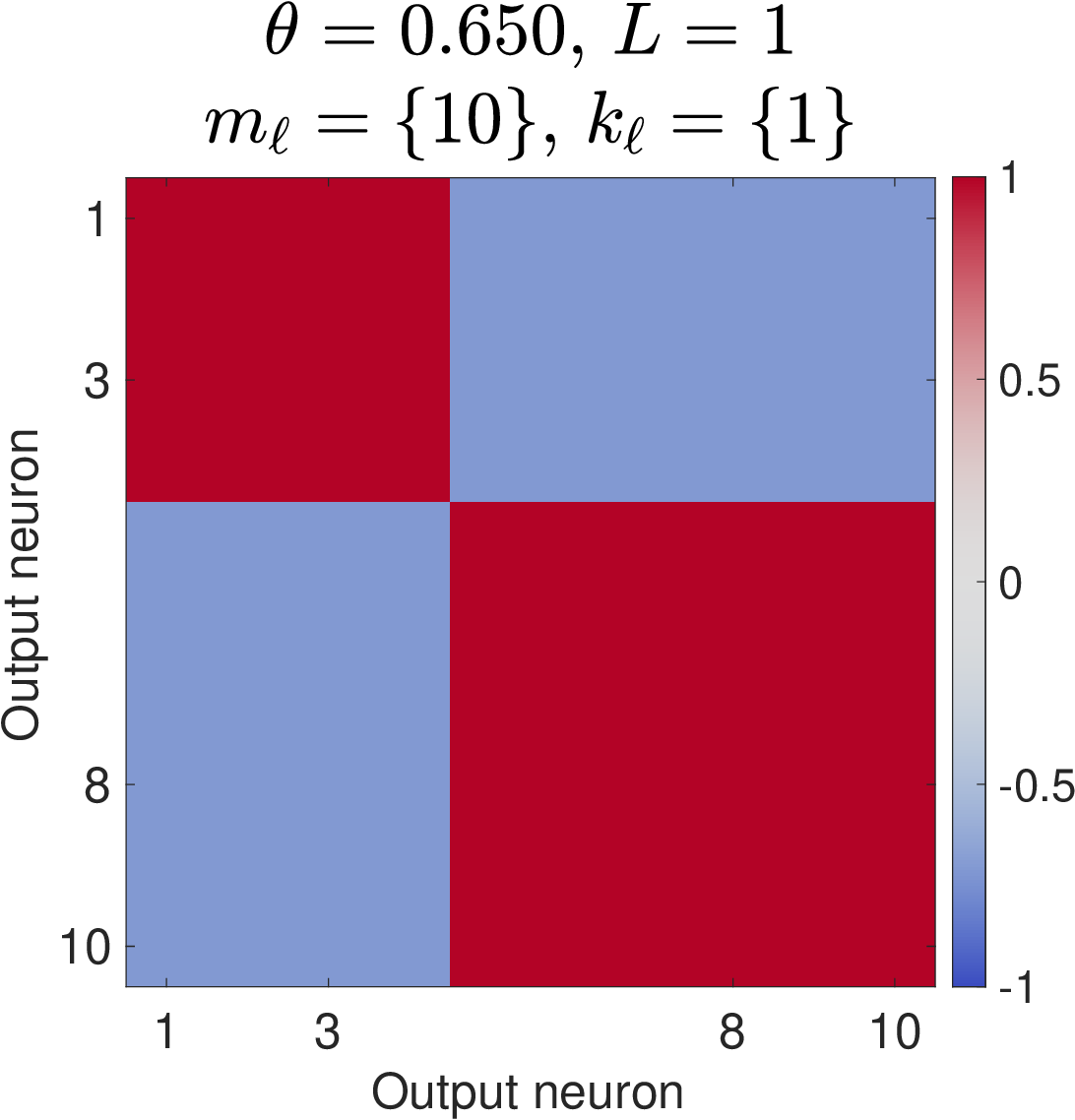}
    \end{subfigure}
    \begin{subfigure}[t]{0.4\textwidth}
        \centering
        \includegraphics[height=5.2cm]{ 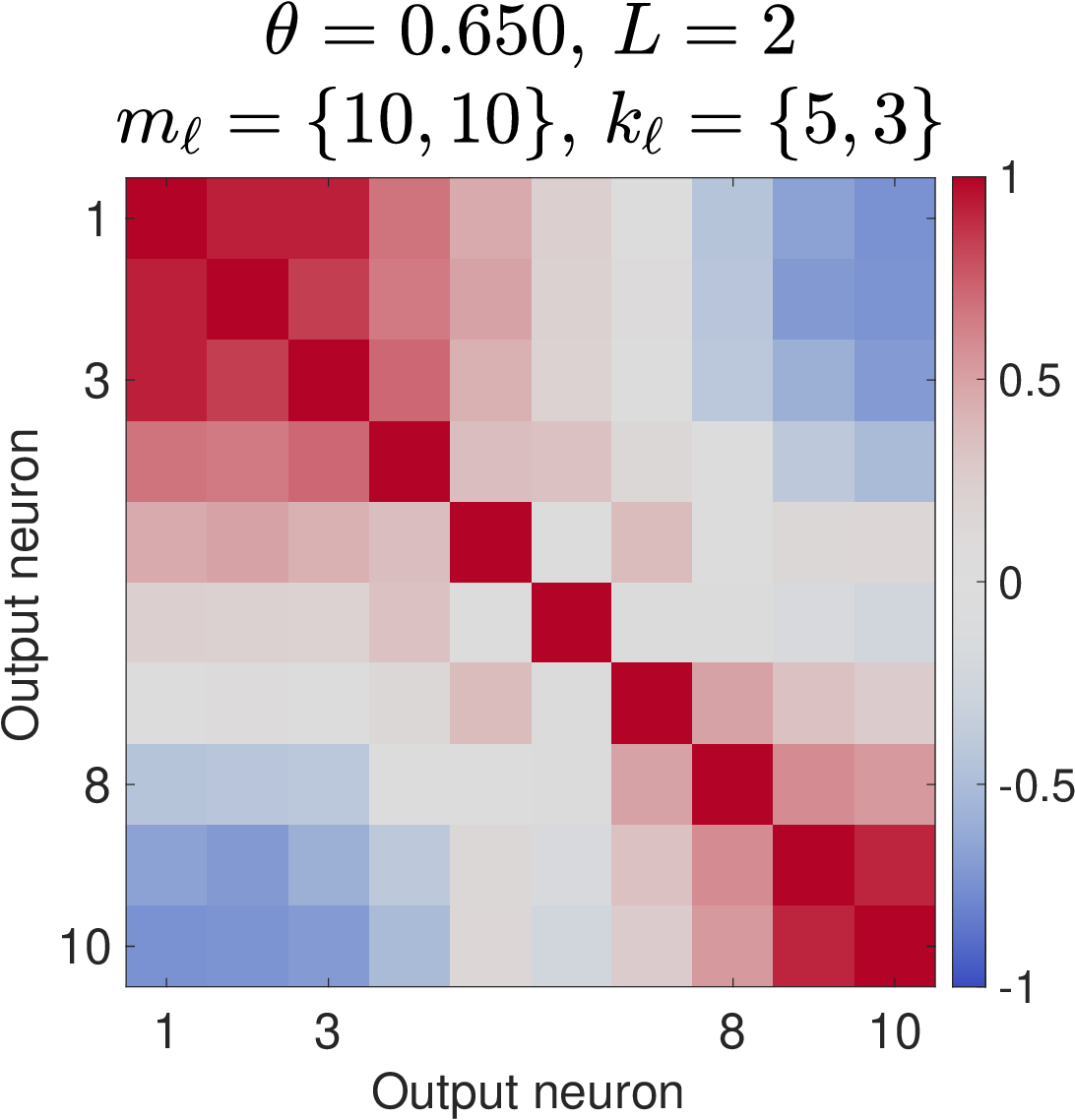}
    \end{subfigure}

    \vspace{0.5em}

    \begin{subfigure}[t]{0.4\textwidth}
        \centering
        \includegraphics[height=5.2cm]{ 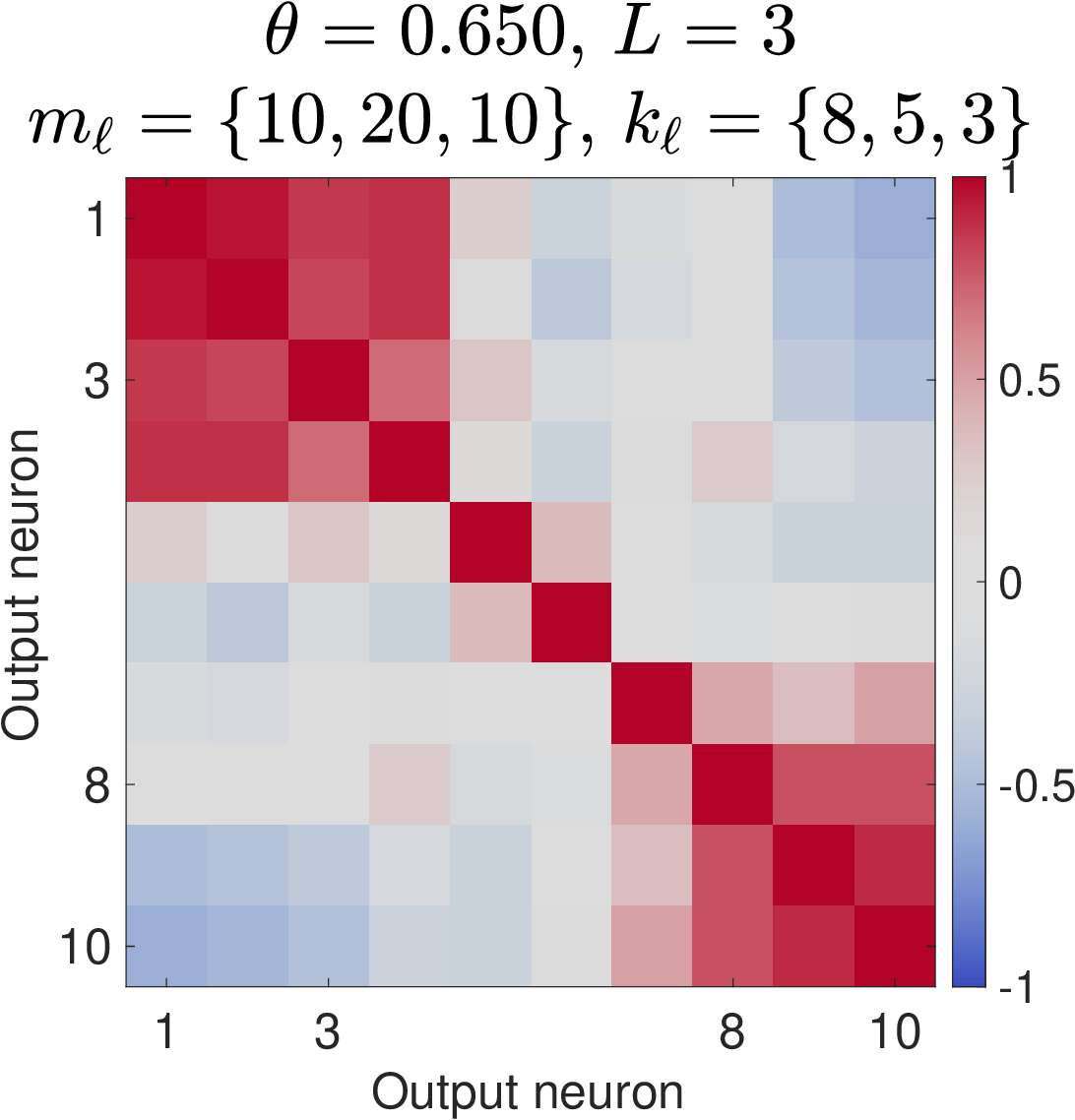}
    \end{subfigure}
    \begin{subfigure}[t]{0.4\textwidth}
        \centering
        \includegraphics[height=5.2cm]{ 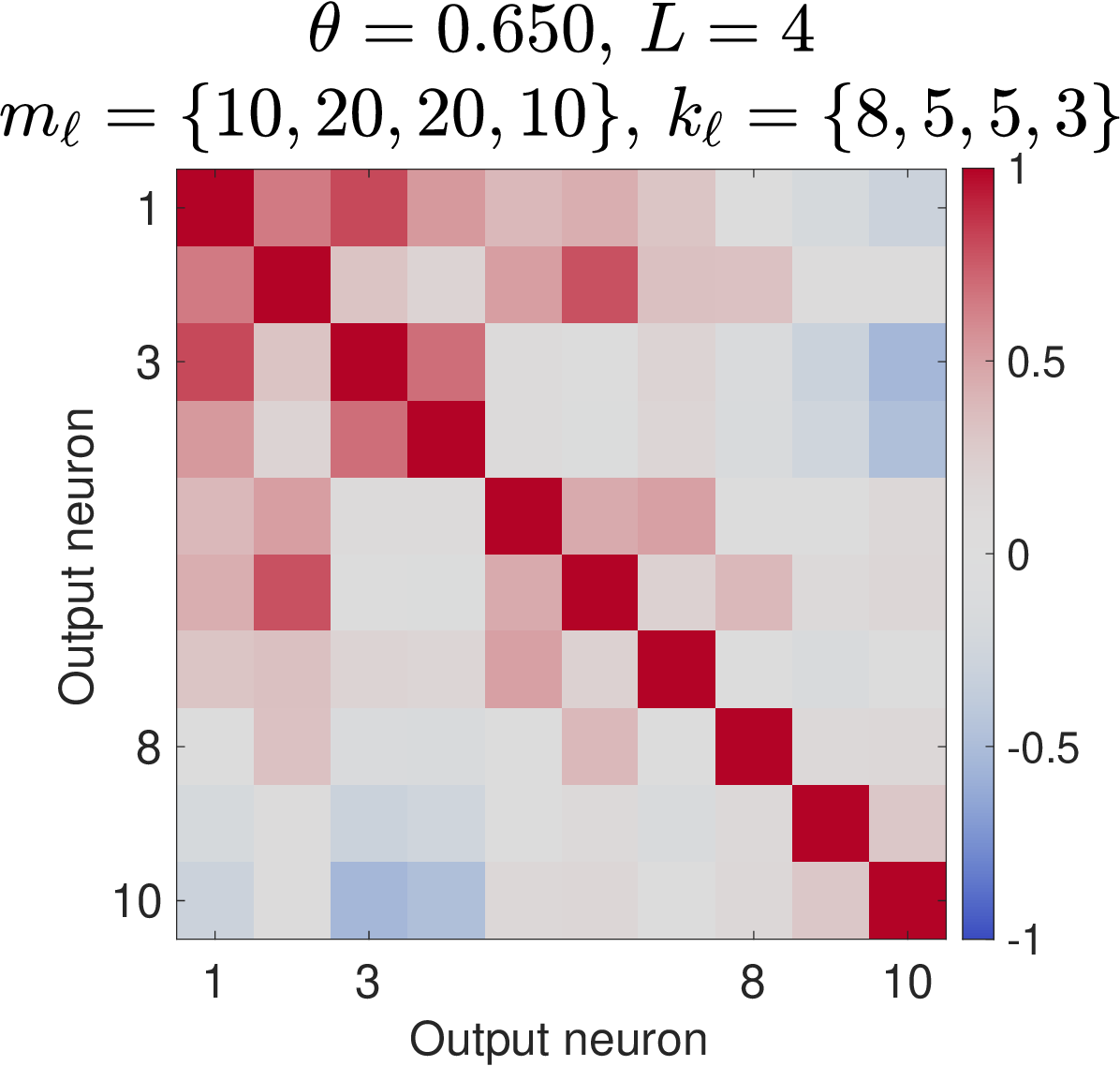}
    \end{subfigure}

    \caption{Empirical correlation matrices with dependence parameter  $\theta=0.65$ in the Gaussian AR(1) input time series. The numbering of the output neurons has been chosen to maximize the correlation of successive neurons.}
    \label{fig:settingB_all}
\end{figure}

\begin{figure}[t]
\centering
\begin{minipage}{0.48\linewidth}
    \centering
    \includegraphics[width=\linewidth]{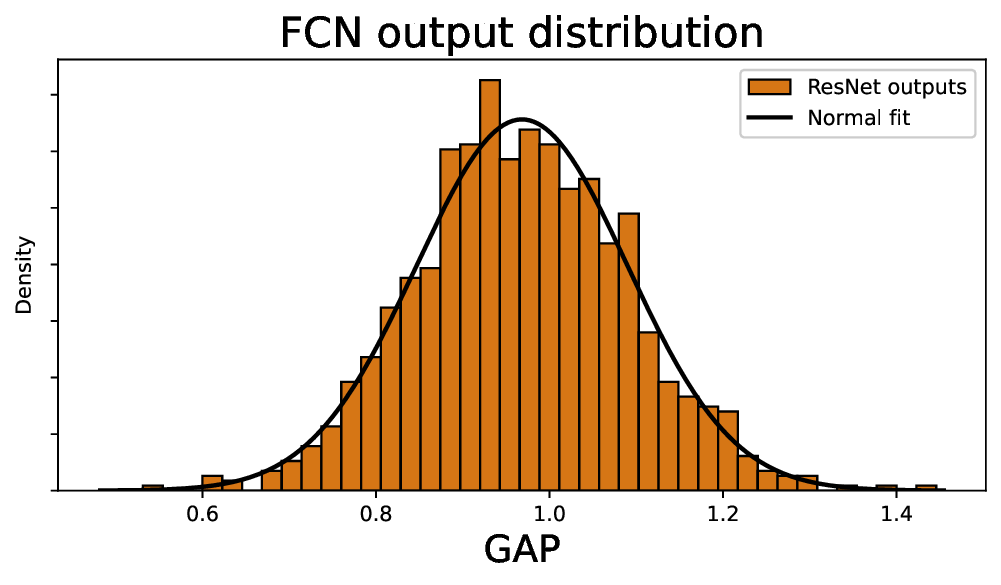}
    \includegraphics[width=0.6\linewidth]{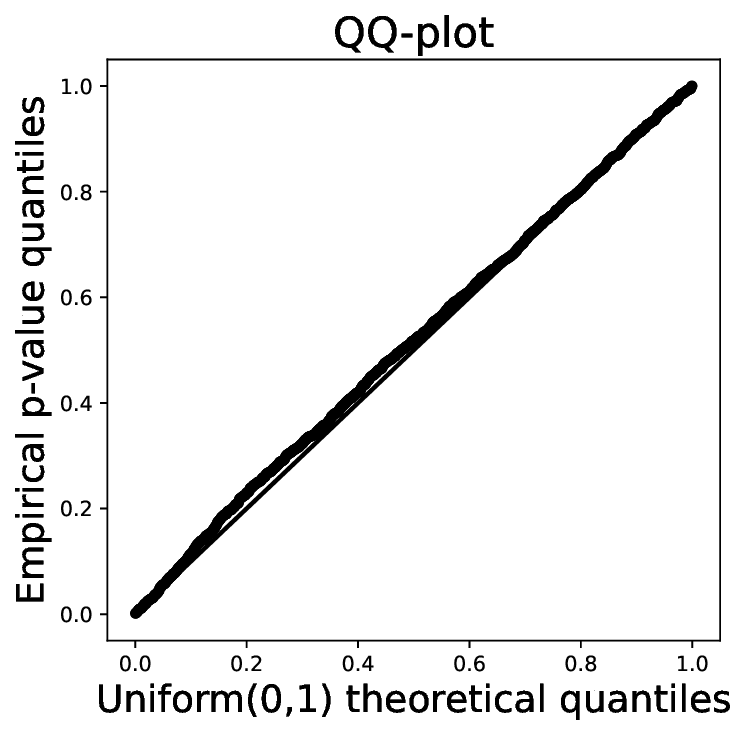}
    \subcaption{$n=1000$, $B=4$, MA(1) input}
\end{minipage}\hfill
\begin{minipage}{0.48\linewidth}
    \centering
    \includegraphics[width=\linewidth]{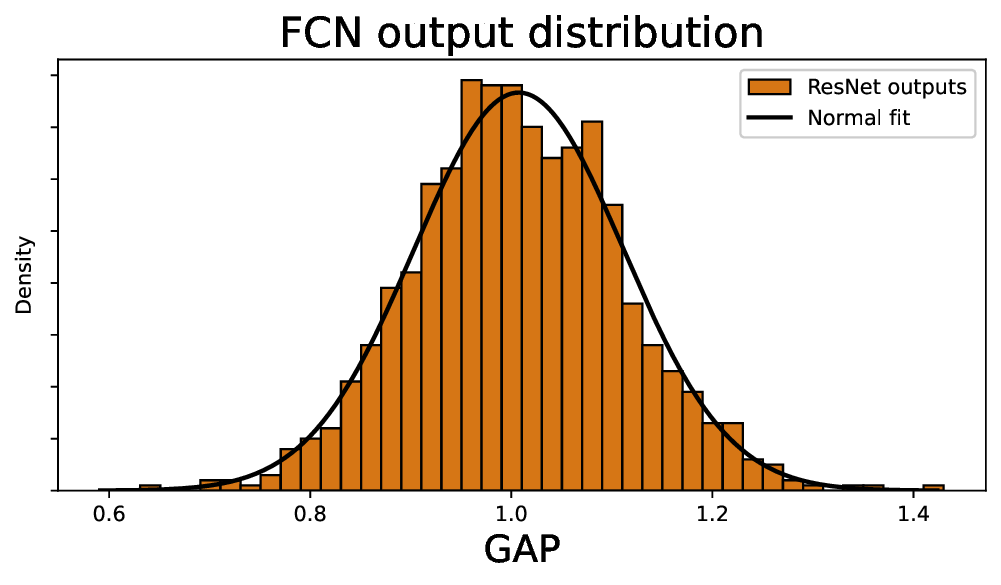}
    \includegraphics[width=0.6\linewidth]{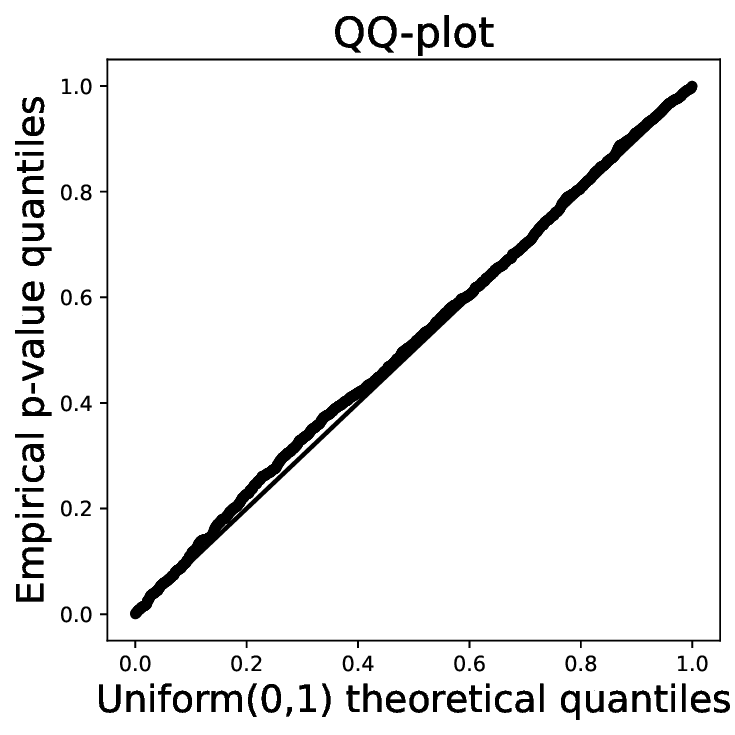}
    \subcaption{$n=1000$, $B=4$, AR(1)}
\end{minipage}
\caption{QQ plots and histograms (with fitted Gaussian densities) of the standardized statistic $S_n$ for two representative configurations.}
\label{fig:simulation_histogram}
\end{figure}

\subsection{Empirical comparison of GAP and weighted-GAP }
\label{sec:empirical_wgap_lambda1}

We compare standard global average pooling (GAP) with the proposed  weighted-GAP (WGAP). 

We consider one-dimensional input time series of length $n$ and FCNs without residual connections ($I^{(\ell)}=0$ for all $\ell$),  ReLU activation function, filter lengths $(k_1,k_2,k_3)=(8,5,3)$, and channel widths $(m_1,m_2,m_3)=(128,256,128).$
Batch normalization is applied.

We then either use a GAP layer or a weighted GAP layer and finally pass the signal through a dense layer \eqref{eq.dense_layer} and a softmax layer \eqref{eq.soft_max} to obtain a probability vector.

Consider a labeled sample $(\mathbf x,y)$ with class label $y\in\{1,\ldots,K\}.$ Given parameter vector $\bm{\theta}$, assume that the model maps the inputs $\mathbf{x}$ to the probability vector $(p_1(\mathbf x, \bm{\theta}), \ldots, p_K(\mathbf x, \bm{\theta})).$ Then, the cross-entropy loss of this sample is $\ell( \mathbf{x},y)
:=
-\log\!(p_{y}(\mathbf x, \bm{\theta})),$ penalized with the first-order difference penalty \eqref{eq.penalty} in the case of WGAP. In the experiments we employ the Adam optimizer with learning rate $10^{-3}$. The regularization parameter $\lambda$ is chosen based on 5-fold cross-validation. 

We test three methods, (i) FCNs with GAP layer, (ii) FCNs with WGAP layer and no regularization, (iii) FCNs with WGAP layer and first-order difference penalty. 

Each dataset comes with a predefined split into training data and  test data (as provided by the
UCR/UEA archive, \cite{UCR}); all reported test accuracies are computed on the corresponding test split.
For each training run, corresponding to one dataset, the network is initialized using PyTorch’s default parameter initialization.
The results can be found in Table \ref{tab:completed_dataset_attempts}. We regard two methods as \emph{comparable} on a given dataset whenever their average test accuracies differ by at most \(2\%\). Under this criterion, penalized WGAP (\(\mathrm{reg\_WGAP}\)) is comparable to or better than GAP on \(29\) out of \(39\) datasets. By contrast, GAP outperforms penalized WGAP by more than two percentage points on the remaining \(10\) datasets. Moreover, penalized WGAP ranks either first or second on \(37\) of the \(39\) datasets.
The code used for the empirical comparison of the three methods can be found at the public available repository:

\url{https://github.com/george24GM/Regularized-GAP-for-time-series-classification.git}

\begin{table}[p]
\centering
\scriptsize
\setlength{\tabcolsep}{6pt}
\caption{Completed simulation by dataset. For each dataset, the best mean accuracy among GAP, WGAP, and reg\_WGAP is highlighted in red, while the second-best is highlighted in orange. The last column reports the median selected regularization parameter \(\lambda\) across the five runs. Each method is trained for \(200\) epochs using the same batch size, while the regularized variant reg\_WGAP selects \(\lambda\) by \(5\)-fold cross-validation, with \(80\) training epochs in each fold.}
\label{tab:completed_dataset_attempts}
\resizebox{\textwidth}{!}{%
\begin{tabular}{lcccc}
\toprule
Dataset & GAP & WGAP & reg\_WGAP & Median \(\lambda\) \\
\midrule
Adiac & 0.53 $\pm$ 0.12 & \textcolor{secondorange}{0.77 $\pm$ 0.03} & \textcolor{bestred}{0.80 $\pm$ 0.01} & 1.024 \\
ArrowHead & \textcolor{bestred}{0.79 $\pm$ 0.03} & 0.76 $\pm$ 0.02 & \textcolor{secondorange}{0.78 $\pm$ 0.03} & 0.128 \\
BeetleFly & \textcolor{bestred}{0.78 $\pm$ 0.06} & 0.67 $\pm$ 0.08 & \textcolor{secondorange}{0.72 $\pm$ 0.13} & 0.032 \\
BirdChicken & \textcolor{bestred}{0.92 $\pm$ 0.03} & 0.72 $\pm$ 0.04 & \textcolor{secondorange}{0.80 $\pm$ 0.07} & 1.024 \\
Car & 0.63 $\pm$ 0.23 & \textcolor{secondorange}{0.76 $\pm$ 0.05} & \textcolor{bestred}{0.81 $\pm$ 0.08} & 0.512 \\
Meat & 0.59 $\pm$ 0.17 & \textcolor{secondorange}{0.86 $\pm$ 0.04} & \textcolor{bestred}{0.92 $\pm$ 0.01} & 2.048 \\
Ham & 0.67 $\pm$ 0.09 & \textcolor{secondorange}{0.70 $\pm$ 0.05} & \textcolor{bestred}{0.77 $\pm$ 0.02} & 0.256 \\
FordB & \textcolor{secondorange}{0.87 $\pm$ 0.04} & 0.81 $\pm$ 0.01 & \textcolor{bestred}{0.90 $\pm$ 0.00} & 32.768 \\
FordA & \textcolor{bestred}{0.91 $\pm$ 0.01} & 0.85 $\pm$ 0.01 & \textcolor{secondorange}{0.91 $\pm$ 0.01} & 32.768 \\
FISH & 0.75 $\pm$ 0.10 & \textcolor{secondorange}{0.90 $\pm$ 0.02} & \textcolor{bestred}{0.91 $\pm$ 0.02} & 1.024 \\
FaceAll & \textcolor{secondorange}{0.91 $\pm$ 0.02} & \textcolor{bestred}{0.93 $\pm$ 0.03} & 0.90 $\pm$ 0.05 & 16.384 \\
Earthquakes & 0.72 $\pm$ 0.20 & \textcolor{secondorange}{0.76 $\pm$ 0.02} & \textcolor{bestred}{0.82 $\pm$ 0.00} & 16.384 \\
CBF & \textcolor{bestred}{0.99 $\pm$ 0.00} & 0.98 $\pm$ 0.01 & \textcolor{secondorange}{0.99 $\pm$ 0.01} & 0.001 \\
Beef & 0.47 $\pm$ 0.10 & \textcolor{secondorange}{0.76 $\pm$ 0.05} & \textcolor{bestred}{0.83 $\pm$ 0.06} & 0.256 \\
ChlorineConcentration & 0.61 $\pm$ 0.05 & \textcolor{secondorange}{0.83 $\pm$ 0.01} & \textcolor{bestred}{0.83 $\pm$ 0.01} & 0.002 \\
CinC\_ECG\_torso & 0.65 $\pm$ 0.09 & \textcolor{secondorange}{0.85 $\pm$ 0.05} & \textcolor{bestred}{0.86 $\pm$ 0.05} & 0.032 \\
Coffee & \textcolor{bestred}{1.00 $\pm$ 0.00} & \textcolor{bestred}{1.00 $\pm$ 0.00} & \textcolor{bestred}{1.00 $\pm$ 0.00} & 0.001 \\
Computers & \textcolor{bestred}{0.82 $\pm$ 0.03} & 0.69 $\pm$ 0.05 & \textcolor{secondorange}{0.77 $\pm$ 0.04} & 8.192 \\
Cricket\_Y & \textcolor{secondorange}{0.69 $\pm$ 0.07} & 0.68 $\pm$ 0.02 & \textcolor{bestred}{0.75 $\pm$ 0.01} & 4.096 \\
Cricket\_Z & \textcolor{bestred}{0.74 $\pm$ 0.03} & 0.63 $\pm$ 0.05 & \textcolor{secondorange}{0.69 $\pm$ 0.05} & 16.384 \\
DiatomSizeReduction & 0.93 $\pm$ 0.00 & \textcolor{secondorange}{0.96 $\pm$ 0.02} & \textcolor{bestred}{0.96 $\pm$ 0.02} & 0.001 \\
DistalPhalanxOutlineAgeGroup & 0.75 $\pm$ 0.16 & \textcolor{secondorange}{0.81 $\pm$ 0.03} & \textcolor{bestred}{0.86 $\pm$ 0.00} & 0.128 \\
DistalPhalanxOutlineCorrect & 0.71 $\pm$ 0.17 & \textcolor{secondorange}{0.78 $\pm$ 0.02} & \textcolor{bestred}{0.82 $\pm$ 0.00} & 16.384 \\
DistalPhalanxTW & \textcolor{secondorange}{0.79 $\pm$ 0.00} & 0.77 $\pm$ 0.01 & \textcolor{bestred}{0.81 $\pm$ 0.01} & 4.096 \\
ECG5000 & \textcolor{secondorange}{0.94 $\pm$ 0.00} & 0.93 $\pm$ 0.00 & \textcolor{bestred}{0.94 $\pm$ 0.00} & 0.008 \\
ECGFiveDays & \textcolor{bestred}{0.99 $\pm$ 0.01} & 0.80 $\pm$ 0.05 & \textcolor{secondorange}{0.87 $\pm$ 0.09} & 4.096 \\
FaceFour & \textcolor{bestred}{0.93 $\pm$ 0.01} & 0.76 $\pm$ 0.05 & \textcolor{secondorange}{0.78 $\pm$ 0.05} & 0.002 \\
FacesUCR & 0.90 $\pm$ 0.05 & \textcolor{secondorange}{0.91 $\pm$ 0.01} & \textcolor{bestred}{0.93 $\pm$ 0.01} & 32.768 \\
Strawberry & 0.93 $\pm$ 0.07 & \textcolor{secondorange}{0.97 $\pm$ 0.01} & \textcolor{bestred}{0.98 $\pm$ 0.00} & 0.256 \\
RefrigerationDevices & \textcolor{secondorange}{0.48 $\pm$ 0.03} & 0.47 $\pm$ 0.01 & \textcolor{bestred}{0.56 $\pm$ 0.04} & 32.768 \\
ScreenType & \textcolor{secondorange}{0.59 $\pm$ 0.06} & 0.48 $\pm$ 0.03 & \textcolor{bestred}{0.60 $\pm$ 0.04} & 16.384 \\
SwedishLeaf & \textcolor{secondorange}{0.96 $\pm$ 0.01} & 0.96 $\pm$ 0.01 & \textcolor{bestred}{0.96 $\pm$ 0.01} & 0.512 \\
Symbols & \textcolor{bestred}{0.95 $\pm$ 0.00} & 0.89 $\pm$ 0.06 & \textcolor{secondorange}{0.90 $\pm$ 0.07} & 0.128 \\
synthetic\_control & 0.98 $\pm$ 0.00 & \textcolor{secondorange}{0.99 $\pm$ 0.00} & \textcolor{bestred}{0.99 $\pm$ 0.00} & 0.001 \\
Wine & 0.50 $\pm$ 0.00 & \textcolor{secondorange}{0.73 $\pm$ 0.14} & \textcolor{bestred}{0.88 $\pm$ 0.09} & 0.008 \\
WordsSynonyms & 0.46 $\pm$ 0.02 & \textcolor{secondorange}{0.56 $\pm$ 0.03} & \textcolor{bestred}{0.58 $\pm$ 0.03} & 0.002 \\
Worms & \textcolor{bestred}{0.56 $\pm$ 0.04} & 0.48 $\pm$ 0.03 & \textcolor{secondorange}{0.55 $\pm$ 0.01} & 4.096 \\
WormsTwoClass & \textcolor{bestred}{0.75 $\pm$ 0.02} & 0.57 $\pm$ 0.02 & \textcolor{secondorange}{0.61 $\pm$ 0.04} & 0.256 \\
yoga & 0.72 $\pm$ 0.11 & \textcolor{secondorange}{0.73 $\pm$ 0.03} & \textcolor{bestred}{0.75 $\pm$ 0.02} & 4.096 \\
\bottomrule
\end{tabular}
}
\end{table}

\section{Related literature}
The study of the asymptotic distribution of weighted sums of stationary processes has been extensively developed by \cite{abadir2014asymptotic}, who establish central limit theorems (CLTs) for weighted sums $\frac{1}{\sqrt{n}}S_n = \sum_{t=1}^n z_{nt} X_t$, where $(z_{nt})$ forms a triangular array, and  $(X_t)_{t\geq 1}$ is a linear process with martingale difference or weakly dependent innovations. Their Theorem 2.2 gives minimal and easily verifiable conditions for asymptotic normality, applicable to short-memory linear processes. \cite{dalla2014studentizing} extend these results, developing robust methods that remain valid without requiring knowledge of the dependence structure.

Our work builds on and generalizes these results by allowing for subordination; that is, we consider a series of random variables of the form $S_n=\sum_{t=1}^n a_t^{(n)} G(X_t)$,  as opposed to a fixed transformation $G$, i.e $S_n=\sum_{t=1}^n G(X_t)$.
Moreover, as a special case, our results corresponds to 
Theorem 2.2 of \cite{abadir2014asymptotic} when $f(x) = x$, the weights satisfy standard summability, and the process has short memory. We also refer to the survey by \cite{merlevede2006recent} for a broader perspective on invariance principles and CLTs for subordinated stationary processes.

CNNs have also been successfully employed to infer structural properties of stochastic time series models.
In particular, \cite{TANG2021113544}
 propose a CNN architecture for ARMA model order identification, trained entirely on simulated time series with known ground truth. Their networks take raw time series as input and learn to predict the autoregressive and moving-average orders directly, significantly outperforming classical likelihood-based criteria such as AIC and BIC both in accuracy and computational efficiency. From a methodological perspective, their CNN architecture processes the time series through stacked one-dimensional convolutions and nonlinearities, followed by a global aggregation over temporal locations before the final decision layer.

The derived result is conditional on the training data. From an ML perspective this makes sense, as we are given a trained model and now deploy it during test time. Obtaining an unconditional statement is much harder. Interestingly, the popular ROCKET method \cite{Dempster2020} works with random filter weights in the convolutional layers that are not trained. It is conceivable that independently drawn random parameters can be included into the analysis.

A possible future direction is to study the asymptotic activation distribution of transformer architectures if the inputs are time series, \cite{Ahmed2023}.

\nocite{*}

\bibliography{bib}

\newpage
\appendix
\onecolumn
\begin{center}
    \Large{\textbf{Appendix}}
\end{center}

\subsection*{Additional Notation}

The symbol $\mathcal{C}$ denotes a generic constant that may vary between occurrences. Let $$\mathcal{F}_t = \sigma(\mathbf{Z}_t, \mathbf{Z}_{t-1}, \ldots)$$ denote the $\sigma$-field generated by the past values of the innovations $\mathbf{Z}_t$.  
Then, for a function $f:\mathbb{R}^k\longrightarrow \mathbb{R}$,  and $\mathbf{X}_t=\sum_{j=0}^\infty A_j \mathbf{Z}_{t-j}$ the following martingale difference decomposition holds
\begin{align}
\nonumber
f(\mathbf{X}_t) - \mathbb{E}[f(\mathbf{X}_t)]
&= \sum_{j=0}^\infty U_{t,j}, 
\quad \text{where} \ \ 
U_{t,j}:= 
\mathbb{E}\big[f(\mathbf{X}_t) \mid \mathcal{F}_{t-j}\big]
- \mathbb{E}\big[f(\mathbf{X}_t) \mid \mathcal{F}_{t-j-1}\big]; 
\label{eq:um}
\end{align}
see  \cite{ho1997limit}.
 
For negative integers $h$, we set $G_{m\ell,n}(h):=G_{m\ell,n}(-h).$  
For $x,y\in \mathbb{R}$,  $$ x\vee y:=\max\{x,y\}.$$

\section{Proof of Theorem~\ref{thm:main} }
\label{sec:proof_multivariate}

In what follows, we show that under the assumptions of
Theorem~\ref{thm:main}, the statistic
\[
S_n
\;=\;
\sum_{t=1}^n \sum_{m=1}^M
a_{t,m}^{(n)}\big(f_m(\mathbf X_t)-\E[f_m(\mathbf X_t)]\big)
\]
satisfies the central limit theorem
\begin{equation}
\frac{1}{\sqrt{n}}\,S_n
\;\xrightarrow{\mathcal D}\;
\mathcal N(0,\sigma^2),
\label{eq:M_CLT}\tag{M--CLT}
\end{equation}
where the asymptotic variance $\sigma^2$ is given by
\begin{align*}
\sigma^2
&=
\sum_{m,\ell=1}^M
\Bigg(
G_{m\ell}(0)\,
\Cov\!\big(f_m(\mathbf X_1),f_\ell(\mathbf X_1)\big)
+
2\sum_{h=1}^\infty
G_{m\ell}(h)\,
\Cov\!\big(f_m(\mathbf X_1),f_\ell(\mathbf X_{1+h})\big)
\Bigg)
<\infty.
\end{align*}
Once~\eqref{eq:M_CLT} is established, the assertion of
Theorem~\ref{thm:main} follows immediately by the Cramér--Wold device.
Indeed, for any fixed vector
$\bm{\lambda}=(\lambda_1,\ldots,\lambda_M)^\top\in\mathbb{R}^M$ and $\mathbf S_n:=(S_{1, n},\ldots, S_{M, n})^\top$ with $S_{m,n}:=\sum_{t=1}^n
\lambda_m a^{(n)}_{t,m}
\big(f_m(\mathbf X_t)-\mathbb{E}[f_m(\mathbf X_t)]\big),$
\[
\bm{\lambda}^\top \mathbf S_n
=
\sum_{m=1}^M \lambda_m S_{m,n}
=
\sum_{t=1}^n
\sum_{m=1}^M
\lambda_m a^{(n)}_{t,m}
\big(f_m(\mathbf X_t)-\mathbb{E}[f_m(\mathbf X_t)]\big).
\]
Define modified weights $\tilde a^{(n)}_{t,m}:=\lambda_m a^{(n)}_{t,m}$.
Since conditions \eqref{eq:W1-multi} and \eqref{eq:W2-multi} hold for $\tilde a^{(n)}_{t,m}$, with limiting coefficients
$\lambda_m\lambda_\ell\,G_{m\ell}(h)$, \eqref{eq:M_CLT} yields
\[
\frac{1}{\sqrt{n}}\,\bm{\lambda}^\top\mathbf S_n
\;\xrightarrow{\mathcal D}\;
\mathcal N\!\big(0, \tilde{\sigma}^2\big), \ \text{with} \ \tilde{\sigma}^2=\bm{\lambda}^\top\Sigma\,\bm{\lambda}.
\]
Since this convergence holds for every $\bm{\lambda}\in\mathbb{R}^M$,
the claim follows from the Cramér--Wold device. 

\textbf{Proof of \eqref{eq:M_CLT}}
 Let $f_m:\mathbb{R}^k\to\mathbb{R}$ be Lipschitz continuous and let the weights $(a_{t,m}^{(n)})$ satisfy conditions \eqref{eq:W1-multi}–\eqref{eq:W2-multi} of Theorem~\ref{thm:main}.  
For each $m\in\{1,\dots,M\}$, define $(U_{t,j,m})_{j\ge 0}$, $$ U_{t,j,m}
:=\E\big[f_m(\mathbf{X}_t)\mid\mathcal{F}_{t-j}\big]-\E\big[f_m(\mathbf{X}_t)\mid\mathcal{F}_{t-j-1}\big],$$ such that
\[
f_m(\mathbf X_t)-\E[ f_m(\mathbf X_t)]
=
\sum_{j=0}^{\infty} U_{t,j, m},
\qquad t\ge 1\;.
\]
With
\[
W_{t,j}^{(n)}
:=\sum_{m=1}^M a_{t,m}^{(n)} U_{t,j,m}\;
\]
it follows that
\begin{align*}
\frac{1}{\sqrt{n}}S_n
&=\frac{1}{\sqrt n}
\sum_{t=1}^n\sum_{m=1}^M a_{t,m}^{(n)}
\big(f_m(\mathbf X_t)-\E f_m(\mathbf X_t)\big) =\frac{1}{\sqrt n}
\sum_{t=1}^n\sum_{j=0}^\infty W_{t,j}^{(n)}.
\end{align*}

We establish \eqref{eq:M_CLT}
through an application of the following theorem: 
\begin{theorem}[Theorem 4.2  in \cite{billingsley:1968}, p. 25]
    \label{theorem_4.2}
     If there exist processes $S_n$ and $V_{u,n}$ and $\sigma_u>0,$ such that for all $u=1,2,\ldots$,
    \begin{itemize}
        \item [(i)] $V_{u,n} \overset{\mathcal{D}}{\longrightarrow}  \mathcal{N}(0,\sigma_u ^2),$
        \item [(ii)] $\sigma^2 =\lim_{u\to \infty} \sigma_u^2 >0$ exists and is finite,
        \item [(iii)] 
        $ \liminf \limits_{u\to \infty} \, \limsup\limits_{n\to \infty} \, \P ( |V_{u,n}-\tfrac{1}{\sqrt{n}}S_n|\geq \varepsilon)=0,$ for any $\varepsilon>0\;,$
    \end{itemize}
    Then, $n^{-1/2}S_n \overset{\mathcal{D}}{\longrightarrow}\mathcal{N}(0,\sigma^2)$ as $n\to \infty.$
\end{theorem}

For this, we define
\begin{equation}
\label{eq:Vu-def-multi}
V_{u,n}
:=
\frac{1}{\sqrt n}
\sum_{t=1}^n\sum_{j=0}^{u-1} W_{t,j}^{(n)}
\end{equation}
and verify conditions $(i), (ii), (iii)$ for $V_{u,n}$ and $S_n$.

\medskip

\paragraph{Proof of $(i)$}
For fixed $u\in\mathbb{N}$ we define, analogously to \cite{Furmanczyk},
\[
M_{t,n}(u)
:=
\frac{1}{\sqrt n}
\sum_{j=0}^{u-1} W_{t+j,j}^{(n)}
=
\frac{1}{\sqrt n}
\sum_{j=0}^{u-1}\sum_{m=1}^M 
a_{t+j,m}^{(n)}\,U_{t+j,j,m},
\qquad t=1,\dots,n.
\]
Reindexing as in \cite{Furmanczyk} shows that
\begin{equation}
\label{eq:Vu_decomposition-multi}
V_{u,n}
=
H_{u,n}
+
\sum_{t=1}^n M_{t,n}(u),
\end{equation}
where
\begin{equation}
\label{eq:H_un-multi}
H_{u,n}
:=
\frac{1}{\sqrt n}\sum_{j=0}^{u-1}
\bigg(
\sum_{t=1}^j W_{t,j}^{(n)}
-
\sum_{t=n+1}^{n+j} W_{t,j}^{(n)}
\bigg).
\end{equation}

For each fixed $u$, the sequence $(M_{t,n}(u))_{t\ge 1}$ is a martingale difference sequence with respect to the filtration $(\mathcal F_t)_{t\ge 1}$, satisfying 
\[
M_{t,n}(u)\in L^2,\qquad
M_{t,n}(u)\text{ is }\mathcal F_t\text{-measurable},\qquad
\E[M_{t,n}(u)\mid \mathcal F_{t-1}]=0.
\]
Indeed, $M_{t,n}(u)$ is a finite sum of $L^2$ functions, and each $U_{t+j,j,m}$ is $\mathcal F_{t}$-measurable and satisfies
\(
\E[U_{t+j,j,m}\mid\mathcal F_{t-1}]=0,
\)
hence linearity of the  conditional expectation implies 
that
$M_{t,n}(u)$ is $\mathcal F_t$-measurable
and $\E[M_{t,n}(u)\mid \mathcal F_{t-1}]=0$.  
For establishing (i), we apply the following central limit theorem for martingale difference arrays to $( M_{t,n}(u))_{t\geq 1}$:
\begin{theorem}[Theorem 1 in \cite{Pollard1984}, p. 171]
\label{thm:Pollard}
    Let $(\xi_{n,t})_{t \geq 1}$ be a martingale difference array. If, as $n\to \infty$,
    \begin{itemize}
        \item [(A)]  $\sum_{t=1}^n\mathbb{E}[ \xi_{n,t}^2\,|\, \mathcal{F}_{n,t-1}] \xrightarrow{\mathbb{P}} \sigma^2$, with $\sigma^2$ a positive constant, 
        \item [(B)] for every $\varepsilon >0$ it holds that $\sum_{t=1}^n\mathbb{E}[ \xi_{n,t}^2 \mathbf{1}(|\xi_{n,t}|>\varepsilon)\,|\, \mathcal{F}_{n,t-1}] \xrightarrow{\mathbb{P}} 0$,
    \end{itemize}
    then, it follows that
    $$ \sum_{t=1}^n {\xi_{n,t}}\xrightarrow{\mathcal{D}} \mathcal{N}(0,\sigma^2)\;.$$
\end{theorem}

\noindent
{\em Proof of (A).}
For proving (A) it suffices to show that 
\begin{align}
\label{eq:expectation_converges} \tag{A1}
  \mathbb{E}\!\left[\sum_{t=1}^n \mathbb{E}[ M_{t,n}^2(u) \mid \mathcal{F}_{t-1}]\right]
&\xrightarrow{n\to \infty}\sigma_u^2>0, \\
\label{eq:variance_vanishes}  \tag{A2}
\Var \!\left(  \sum_{t=1}^n \mathbb{E}[ M_{t,n}^2(u) \mid \mathcal{F}_{t-1}] \right)
&\xrightarrow{n\to \infty}0.
\end{align}

\noindent
\textit{Proof of} (A1). Due to stationarity of $U_{t,j,m}$ it holds that
\begin{align*}
  \mathbb{E}\!\left[\sum_{t=1}^n \mathbb{E}[ M_{t,n}^2(u) \mid \mathcal{F}_{t-1}]\right] =\sum_{t=1}^n \E[M_{t,n}^2(u)]
&=
\frac{1}{n}\sum_{t=1}^n
\sum_{j,k=0}^{u-1}\sum_{m,\ell=1}^M
a_{t+j,m}^{(n)}a_{t+k,\ell}^{(n)}
\E\big[U_{t+j,j,m}U_{t+k,k,\ell}\big]\\
&=
\sum_{j,k=0}^{u-1}\sum_{m,\ell=1}^M
\E\big[U_{1+j,j, m}U_{1+k,k,\ell}\big]\,
\frac{1}{n}\sum_{t=1}^n a_{t+j,m}^{(n)}a_{t+k,\ell}^{(n)}.
\end{align*}
By \eqref{eq:W1-multi} and Lemma \ref{lem:BnM1-multi}, for each fixed $h=k-j$ and pair $(m,\ell)$,
\[
\frac{1}{n}\sum_{t=1}^n a_{t+j,m}^{(n)}a_{t+k,\ell}^{(n)}
=
G_{n,m\ell}(h)+o(1),
\qquad
G_{n,m\ell}(h)
:=
\frac1n\sum_{s=1}^n a_{s,m}^{(n)} a_{s+h,\ell}^{(n)},
\]
and $G_{n,m\ell}(h)\to G_{m\ell}(h)$ as $n\to\infty$. 
It follows that
\begin{equation}
\label{eq:sigmas-multi}
\sum_{t=1}^n \E[M_{t,n}^2(u)]
\xrightarrow[n\to\infty]{}
\sigma_u^2
:=
\sum_{j,k=0}^{u-1}\sum_{m,\ell=1}^M
\E\big[U_{1+j,j,m}U_{1+k,k,\ell}\big]\,
G_{m\ell}(k-j)>0.
\end{equation}

\noindent
\textit{Proof of} (A2).
Recall that
\[
M_{t,n}(u)
=
\frac{1}{\sqrt{n}}\sum_{j=0}^{u-1}\sum_{m=1}^M a_{t+j,m}^{(n)} U_{t+j,j,m}\;.
\]
Define, for $j,k\in\{0,\dots,u-1\}$ and $m,\ell\in\{1,\dots,M\}$,
\[
R_t^{j,k;m,\ell}
:=
\mathbb E\big[U_{t+j,j,m}U_{t+k,k,\ell} \mid \mathcal F_{t-1}\big].
\]
Then, it holds that
\[
\sum_{t=1}^n \mathbb E\big[M_{t,n}^2(u) \mid \mathcal F_{t-1}\big]
=
\frac{1}{n}
\sum_{t=1}^n \sum_{j,k=0}^{u-1} \sum_{m,\ell=1}^M
a_{t+j,m}^{(n)} a_{t+k,\ell}^{(n)} R_t^{j,k;m,\ell}.
\]
Hence, it follows that
\begin{align}
&\Var\Big(\sum_{t=1}^n \mathbb E\big[M_{t,n}^2(u) \mid \mathcal F_{t-1}\big]\Big)
=
\Var\bigg(
\frac{1}{n}
\sum_{t=1}^n \sum_{j,k=0}^{u-1}\sum_{m,\ell=1}^M
a_{t+j,m}^{(n)} a_{t+k,\ell}^{(n)} R_t^{j,k;m,\ell}
\bigg)\nonumber\\
&=
\frac{1}{n^2}
\sum_{s=1}^n \sum_{t=1}^n
\Cov\bigg(
\sum_{j,k=0}^{u-1}\sum_{m,\ell=1}^M
a_{t+j,m}^{(n)} a_{t+k,\ell}^{(n)} R_t^{j,k;m,\ell},\, 
\sum_{r,\eta=0}^{u-1}\sum_{p,q=1}^M
a_{s+r,p}^{(n)} a_{s+\eta,q}^{(n)} R_s^{r,\eta;p,q}
\bigg)\nonumber\\
&=
\frac{1}{n^2}
\sum_{s,t=1}^n
\sum_{j,k,r,\eta=0}^{u-1}
\sum_{m,\ell,p,q=1}^M
a_{t+j,m}^{(n)} a_{t+k,\ell}^{(n)}
a_{s+r,p}^{(n)} a_{s+\eta,q}^{(n)}
\Cov\big(R_t^{j,k;m,\ell},R_s^{r,\eta;p,q}\big)\\
&\leq C_{u,M} \max_{\substack{j,k,r,\eta\\ m,\ell,p,q}} \frac{1}{n^2}
\sum_{s,t=1}^n
a_{t+j,m}^{(n)}a_{t+k,\ell}^{(n)}
a_{s+r,p}^{(n)}a_{s+\eta,q}^{(n)}\,
\Cov\big(R_t^{j,k;m,\ell},R_s^{r,\eta;p,q}\big)\;,
\label{eq:var-full-multi}
\end{align}
where $C_{u,M}$ is a constant that only depends on $u,M$. 
Fix indices $j,k,r,\eta\in\{0,\dots,u-1\}$ and $m,\ell,p,q\in\{1,\dots,M\}$  and define
\[
d_{n,t}
:=
\frac{a_{t+j,m}^{(n)} a_{t+k,\ell}^{(n)}}{n},
\qquad
\tilde d_{n,s}
:=
\frac{a_{s+r,p}^{(n)} a_{s+\eta,q}^{(n)}}{n},
\]
and abbreviate
\[
R_t:=R_t^{j,k;m,\ell},
\qquad
\tilde R_s:=R_s^{r,\eta;p,q}.
\]

For proving \eqref{eq:variance_vanishes} it then suffices to show that
\[
\sum_{s,t=1}^n 
d_{n,t}\,\tilde d_{n,s}\,\Cov(R_t,\tilde{R}_s)\longrightarrow 0, \ \text{as $n\rightarrow \infty$.}
\]
By stationarity
\(
\gamma_0:=\sqrt{\Var(R_1)\Var(\tilde{R}_1)}
=\sqrt{\Var(R_t)\Var(\tilde{R}_s)}.
\)
Set
\[
K_n:= 1\vee \frac{\sqrt{n}}{\max\limits_{\substack{u=1,\ldots,n\\v=1,\ldots,M}} |a_{u,v}^{(n)}|}.
\]
Applying the Cauchy–Schwarz inequality yields
\begin{align*}
\sum_{s,t=1}^n
d_{n,t}\,\tilde{d}_{n,s}\,\Cov(R_t,\tilde{R}_s)
=&
\sum_{\substack{s,t=1\\ |s-t|\le K_n}}^n
d_{n,t}\,\tilde{d}_{n,s}\,\Cov(R_t,\tilde{R}_s)
+
\sum_{\substack{s,t=1\\ |s-t|> K_n}}^n
d_{n,t}\,\tilde{d}_{n,s}\,\Cov(R_t,\tilde{R}_s)
\\[0.5em]
\leq &
\gamma_0 
\sum_{\substack{s,t=1\\ |s-t|\le K_n}}^n
|d_{n,t}\,\tilde{d}_{n,s}|
+
\max_{\substack{s,t=1,\ldots,n\\ |s-t|>K_n}}
|\Cov(R_t,\tilde{R}_{s})|
\sum_{\substack{s,t=1\\ |s-t|> K_n}}^n
|d_{n,t}\,\tilde{d}_{n,s}| \\
=:&I_n+ J_n.
\end{align*}
We now show that both $I_n$ and $J_n$ converge to zero.  
Since $a_{t,\cdot}^{(n)}=0$ for $t\geq n+1$,  
\[
\max_{1\le t\le n}|d_{n,t}|
\le 
\max_{\substack{1\le u\le n\\1\leq r\leq M} }\frac{|a_{u,r}^{(n)}|^2}{n}.
\]
Moreover, by the Cauchy–Schwarz inequality, and since $a_{t,\cdot}^{(n)}=0$ for $t\geq n+1$,
\begin{align}
 \sum_{s=1}^n|\tilde d_{n,s}|
 &=\frac{1}{n}\sum_{s=1}^n |a_{s+r,p}^{(n)}||a_{s+\eta,q}^{(n)}| \notag \\
 &
 \leq 
 \frac{1}{n}
 \sqrt{\sum_{s=1}^n \left(a_{s+r,p}^{(n)}\right)^2
       \sum_{t=1}^n \left(a_{t+\eta,q}^{(n)}\right)^2} \notag \\
 &\leq 
 \sqrt{\frac{1}{n}\sum_{s=1}^n \left(a_{s,p}^{(n)}\right)^2 \frac{1}{n}\sum_{t=1}^n \left(a_{t,q}^{(n)}\right)^2  }. \label{eq:sum_d}
\end{align}
The condition $|s-t|\le K_n$ implies  
\(t - K_n \le s \le t + K_n\).
Thus for each fixed \(t\), the number of such \(s\) is at most \(2K_n+1\). Hence, it holds that
\[
\sum_{\substack{s,t=1\\ |s-t|\le K_n}}^{n} |\tilde d_{n,s}|
=
\sum_{t=1}^{n}
\left(\sum_{\substack{s=1\\ |s-t|\le K_n}}^{n} 1\right)
|\tilde d_{n,t}|
\le
(2K_n+1)\sum_{t=1}^n |\tilde d_{n,t}|.
\]
Since  
\(2K_n+1 \le 3K_n = 3\vee 3\sqrt{n}/\max_{u,r}|a_{u,r}^{(n)}|\),

\begin{align*}
I_n
&\le \gamma_0\max_{1\leq t \leq n} |d_{n,t}| \, \sum_{\substack{s,t=1 \\ |s-t|\le K_n}}  |\tilde d_{n,s}| \\
&\le \gamma_0
\frac{\max\limits_{1\leq u \leq n}\max\limits_{1\leq m \leq M} |a_{u,m}^{(n)}|^2}{n}\,
(2K_n+1)\sum_{t=1}^n |\tilde d_{n,t}| \\[0.8em]
&\le \gamma_0
\left(\frac{3\max\limits_{1\leq u \leq n}\max\limits_{1\leq m \leq M} |a_{u,m}^{(n)}|}{\sqrt{n}}\vee \frac{3\max\limits_{1\leq u \leq n}\max\limits_{1\leq m \leq M} |a_{u,m}^{(n)}|^2}{n}\right)
\cdot
 \sqrt{\frac{1}{n}\sum_{s=1}^n \left(a_{s,p}^{(n)}\right)^2 \frac{1}{n}\sum_{t=1}^n \left(a_{t,q}^{(n)}\right)^2  }\;.
\end{align*}
By \eqref{eq:W1-multi} and \eqref{eq:W2-multi}, this converges to \(0\), so \(I_n\to 0\).
For \(J_n\), again by \eqref{eq:sum_d} and \eqref{eq:W1-multi},
\begin{align*}
    \sum_{\substack{s,t=1\\ |s-t|> K_n}}^n
\left|d_{n,t}\,\tilde{d}_{n,s}\right|
\le 
\sum_{t=1}^n \left|d_{n,t}\right|\;\sum_{s=1}^n\left|\tilde{d}_{n,s}\right|
\end{align*}
is bounded, hence it suffices to show that setting $\gamma_k:=\Cov(R_1, \tilde{R}_{k+1})$ one has
\[
    \max_{\substack{s,t=1\\ |s-t|> K_n}}
      \Cov(R_t,\tilde{R}_{s})
    = \max_{k>K_n} \gamma_k \xrightarrow{n\to \infty}0.
\] 
For this, we consider the following truncated versions of $\mathbf{X}_t$, $U_{t,j,m}$ and $R_{t}^{j,k;m,\ell}$:
\begin{align*}
&\mathbf{X}_t^{(L)}=\sum_{h=0}^L  A_h \mathbf{Z}_{t-h}, 
\\
&U_{t,j,m}^{(L)}
   :=\E\!\left[f_m(\mathbf{X}_t^{(L)}) \mid \mathcal{F}_{t-j}\right]
    - \E\!\big[f_m(\mathbf{X}_t^{(L)}) \mid \mathcal{F}_{t-j-1}\big],  
\\
&R_t^{j,k;m,\ell(L)}
   :=\E\!\left[U_{t+j,j,m}^{(L)}\,U_{t+k,k,\ell}^{(L)}
      \mid \mathcal F_{t-1}\right].
\end{align*}
Due to the Cauchy–Schwarz inequality, the i.i.d.\ property of \(\mathbf{Z}_t\), 
and sub-multiplicativity of the matrix norm,
for any integers \(h,h'\),
\[
\E \big[(A_h\mathbf{Z}_{t-h})^{\top}A_{h'} \mathbf{Z}_{t-h'}\big] 
\le 
\|A_h\|\,\|A_{h'}\|\,
\E[\|\mathbf{Z}_0\|^2].
\]
Thus, using
$\mathbf{X}_t-\mathbf{X}_t^{(L)}
=\sum_{h=L+1}^{\infty}A_h \mathbf{Z}_{t-h}$, it follows that
\[
\E\!\left[\|\mathbf{X}_t-\mathbf{X}_t^{(L)}\|^2\right]
\le 
\E[\|\mathbf{Z}_{0}\|^2]
\left(\sum_{h=L+1}^{\infty}\|A_h\| \right)^2\;.
\]
The right-hand side of the above inequality tends to 0 as $L\rightarrow \infty$ due to the assumption $\sum_{h=0}^{\infty}\|A_h\|<\infty$.
For any $\sigma$-field $\mathcal{G}$, and Lipschitz-continuous  $f$  with Lipschitz-constant $\text{Lip}(f)$, it then follows that 
\begin{align}    \left|\E(f(\mathbf{X}_t)|\mathcal{G})- \E(f(\mathbf{X}_t^{(L)})|\mathcal{G})\right|
\leq \E\left[\left|f(\mathbf{X}_t)-f(\mathbf{X}_t^{(L)})\right|\left|\right.\mathcal{G}\right]\leq \text{Lip}(f)\E\left[\|\mathbf{X}_t-\mathbf{X}_t^{(L)}\|\, \Big|\mathcal{G}\right]\;.
\label{eq:truncated_lip} \tag{T-Lip}
\end{align}
It follows that
\begin{align*}
  \left|U_{t,j,m}-U_{t,j,m}^{(L)}\right|  
  \leq \text{Lip}(f_m)\left(\E\left[\|\mathbf{X}_t-\mathbf{X}_t^{(L)}\|\, \Big|\mathcal{F}_{t-j}\right]+\E\left[\|\mathbf{X}_t-\mathbf{X}_t^{(L)}\|\,\Big|\mathcal{F}_{t-j-1}\right]\right)\;.
\end{align*}
Due to Jensen's inequality and using that $(a+b)^2\leq 2(a^2+b^2)$,
\begin{align}
 \E\!\left[ \left|U_{t,j,m}-U_{t,j,m}^{(L)}\right|^2\right]
&\le  
2\,\mathrm{Lip}^2(f_m)\,\E\!\left[
     \big(\E[\|\mathbf{X}_t-\mathbf{X}_t^{(L)}\|\mid\mathcal{F}_{t-j}]\big)^2
     + \big(\E[\|\mathbf{X}_t-\mathbf{X}_t^{(L)}\|\mid\mathcal{F}_{t-j-1}]\big)^2
   \right] \nonumber\\
&\le  
4\,\mathrm{Lip}^2(f_m)\,\E\!\left[\|\mathbf{X}_t-\mathbf{X}_t^{(L)}\|^2\right]
\le 
4\,\mathrm{Lip}^2(f_m)\,\E[\|\mathbf{Z}_0\|^2]
\Big(\sum_{h>L}\|A_h\|\Big)^2.
\label{truncation bound}
\end{align}
Moreover, 
\begin{align*}
R_t^{j,k;m,\ell}-R_t^{j,k;m,\ell,(L)}
=&
\E\!\left[(U_{t+j,j,m}-U_{t+j,j,m}^{(L)})\,U_{t+k,k,\ell}\mid\mathcal F_{t-1}\right]
\;+\;
\E\!\left[U_{t+j,j,m}^{(L)}\,(U_{t+k,k,\ell}-U_{t+k,k,\ell}^{(L)})\mid\mathcal F_{t-1}\right] \\
=:&\mathcal{A}+\mathcal{B}\;,
\end{align*}
and hence
\[
\E\!\left[(R_t^{j,k;m,\ell}-R_t^{j,k,m,\ell,(L)})^2\right]
\le 2\,\E[\mathcal{A}^2] + 2\,\E[\mathcal{B}^2].
\]
By Lemma~\ref{lemma:lemma3.4Furmanczyk} (see inequality~\eqref{eq:almost_sure_bound_last}), 
\(
|U_{t+k,k,\ell}|\le C_\ell\,\|A_k\|\big(\|\mathbf{Z}_{t-1}\|+C\big)
\)
almost surely.  
Since $\mathbf{Z}_{t-1}$ is  $\mathcal F_{t-1}$-measurable, we obtain the almost sure bound
\[
\E[U_{t+k,k,\ell}^2 \mid \mathcal F_{t-1}]
\le C_\ell\,\|A_k\|^2 \big(\|\mathbf{Z}_{t-1}\|+C\big)^2 .
\]
Applying the Cauchy--Schwarz inequality to $\mathcal{A}$ gives
\begin{align*}
|\mathcal{A}|
\le&
\Big(\E[(U_{t+j,j,m}-U_{t+j,j,m}^{(L)})^2\mid\mathcal F_{t-1}]\Big)^{1/2}
\Big(\E[U_{t+k,k,\ell}^2\mid\mathcal F_{t-1}]\Big)^{1/2}\\
\le & 
C_\ell\,\|A_k\|\Big(\E[(U_{t+j,j,m}-U_{t+j,j,m}^{(L)})^2\mid\mathcal F_{t-1}]\Big)^{1/2} \big(\|\mathbf{Z}_{t-1}\|+C\big).
\end{align*}
Squaring, taking expectations, and using the tower property yields
\[
\E[\mathcal{A}^2]
\le 
C_\ell\,\|A_k\|^2\,
\E\!\left[(U_{t+j,j,m}-U_{t+j,j,m}^{(L)})^2\right]\;,
\]
where the constants got absorbed into $C_\ell$.
By \eqref{truncation bound} it then follows that
\[
\E[\mathcal{A}^2]
\le 
C_{m,\ell}\|A_k\|^2\Big(\sum_{h>L}\|A_h\|\Big)^2.
\]
The same reasoning as in Lemma \ref{lemma:lemma3.4Furmanczyk} applies to \(U^{(L)}_{t+j,j,m}\) as well, because the truncated process \( (\mathbf{X}_t^{(L)})_{t\geq 1}\) meets all the hypotheses of Lemma \ref{lemma:lemma3.4Furmanczyk}. This yields
\[
\E[\mathcal{B}^2]
\le 
C_{m,\ell}\,\|A_j\|^2\Big(\sum_{h>L}\|A_h\|\Big)^2.
\]
Note that $C_{m,\ell}$ depends on $\text{Lip}(f_m)$ and $\text{Lip}(f_\ell)$, hence $\max_{1\leq m, \ell \leq M} C_{m,\ell}<\infty$. Moreover, due to finiteness of $\sum_j\|A_j\|<\infty$, there exists $C$ such that, for all $j=0,1,\ldots$,  $\|A_j\|<C$. 
Combining these bounds, we conclude that
\begin{align*}
\E\!\left[(R_t^{j,k;m,\ell}-R_t^{j,k;m,\ell,(L)})^2\right]
\le&
C_{m,\ell}\Big(\sum_{h>L}\|A_h\|\Big)^2\bigl(\|A_j\|^2+\|A_k\|^2\bigr)\\
\le& C_{m,\ell}\Big(\sum_{h>L}\|A_h\|\Big)^2\;,
\end{align*}
which is the key ingredient for controlling  
$\Cov(R_t,\tilde R_s)= \Cov(R_t^{j,k;m,\ell}, R_s^{r,\eta;p,q})$ when $|t-s|$ is large.  Note that
\begin{align}
\label{eq:cov_partial}
 \Cov(R_t^{j,k;m,\ell}, R_s^{r,\eta;p,q})
  =& \Cov\left(R_t^{j,k;m,\ell} - R_t^{j,k;m,\ell,(L)},\, R_s^{r,\eta;p,q}\right) \nonumber\\ 
   &+ \Cov\left(R_t^{j,k;m,\ell,(L)},\, R_s^{r,\eta;p,q} - R_s^{r,\eta;p,q,(L)}\right) \nonumber\\ &+ \Cov\left(R_t^{j,k;m,\ell,(L)},\, R_s^{r,\eta,p,q,(L)}\right) \;.
\end{align} 
Since 
\(
\mathbf{X}^{(L)}_{t+j}
=\sum_{h=0}^{L} A_h \mathbf{Z}_{t+j-h},
\)
the variable \(f_m(\mathbf{X}^{(L)}_{t+j})\) depends only on 
\(\mathbf{Z}_{t+j},\dots,\mathbf{Z}_{t+j-L}\).
All variables $\mathbf{Z}_i$ with index \(i\le t\) are measurable with respect to  \(\mathcal{F}_t\), while those with index \(i>t\) are independent of \(\mathcal{F}_t\).  
Hence, for this quantity,  conditioning on \(\mathcal{F}_t\) is the same as conditioning on the block \((\mathbf Z_t,\dots,\mathbf Z_{t+j-L})\), that means
\[
\mathbb{E}[f_m(\mathbf{X}^{(L)}_{t+j})\mid\mathcal{F}_t]
=
\mathbb{E}\!\left[f_m(\mathbf{X}^{(L)}_{t+j})\mid 
\mathbf Z_t,\dots,\mathbf Z_{t+j-L}\right].
\]
Due to the independence of  $\mathbf{Z}_t$, $t\geq 1$, it follows that
for $L>j$
\begin{align*}
U^{(L)}_{t+j, j,m}=&\mathbb{E}\big[f_m(\mathbf{X}^{(L)}_{t+j}) \mid \mathcal{F}_{t}\big]
- \mathbb{E}\big[f_m(\mathbf{X}^{(L)}_{t+j}) \mid \mathcal{F}_{t-1}\big]\\
=&
\mathbb{E}\big[f_m(\mathbf{X}^{(L)}_{t+j}) \mid \mathbf{Z}_t, \ldots, \mathbf{Z}_{t+j-L}\big]
- \mathbb{E}\big[f_m(\mathbf{X}^{(L)}_{t+j}) \mid \mathbf{Z}_{t-1}, \ldots, \mathbf{Z}_{t+j-L}\big],
\end{align*}
that is, $U^{(L)}_{t+j,j,m}$ is a function of $\mathbf{Z}_{t}, \ldots, \mathbf{Z}_{t+j-L}$ and, analogously, $U^{(L)}_{t+k,k,\ell}$ is a function of $\mathbf{Z}_{t}, \ldots, \mathbf{Z}_{t+k-L}$. 
It follows that $R_t^{j,k,(L)}$ depends only on the innovations 
$\mathbf Z_{t},\ldots,\mathbf Z_{t+\min(j,k)-L}$, while 
$R_s^{r,\ell,(L)}$ depends only on 
$\mathbf Z_{s},\ldots,\mathbf Z_{s+\min(r,\ell)-L}$.  
These two innovation blocks are disjoint (hence the quantities are independent) whenever
\(
t+\min(j,k)-L > s
\quad\text{or}\quad
s+\min(r,\ell)-L > t,
\)
in which case 
$\Cov(R_t^{j,k;m,\ell,(L)},R_s^{r,\eta;p,q,(L)})=0$.  
Therefore, a sufficient condition ensuring non-overlap is
\(
|t-s| > 2L + \max\{j,k,r,\ell\}.
\)

To bound the remaining two summands in \eqref{eq:cov_partial}, we use that the Cauchy–Schwarz inequality and $\Var(X)\leq E[X^2]$ imply 
$|\Cov(U,V)|\le(\mathbb{E}|U-\mathbb{E}[U]|^2)^{1/2}(\mathbb{E}|V-\mathbb{E}[V]|^2)^{1/2}
\le(\mathbb{E}[|U|^2])^{1/2}(\mathbb{E}[|V|^2])^{1/2}$. To apply this bound to $\Cov(R_t^{t,j;m,\ell},R_s^{r,\eta;p,q})$, we first verify that 
$\mathbb{E}[(R_t^{j,k;m,\ell})^2]$ and $\mathbb{E}\left[(R_s^{r,\eta;,p,q})^2\right]$ are uniformly bounded. 
Recall that 
$R_s^{r,\eta;p,q}=\mathbb{E}[U_{s+r,r,p}\,U_{s+\eta,\eta,q}\mid\mathcal{F}_{s-1}]$. 
From \eqref{eq:almost_sure_bound_last} in the proof of Lemma~\ref{lemma:lemma3.4Furmanczyk}, for $ \mathcal{Z}_s:=\|\mathbf{Z}_{s-1}\| + \sqrt{\E[\|\mathbf{Z}_s\|^2]}$, the terms $U_{s+k,k, l}$ satisfy 
the almost sure bound $|U_{s+k,k, l}|\le C_l\|A_k\|\mathcal{Z}_s$, where 
$\mathcal{Z}_s$ depends only on $Z_s$ and satisfies $\mathbb{E}[\mathcal{Z}_s^2]<\infty$. 
It follows that 
\begin{align*}
|R_s^{r,\eta;p,q}|^2
\leq C_{p}C_q \|A_r\|\|A_{\eta}\|
\E\left(\mathcal{Z}_s^2\left|\right. \mathcal{F}_{s-1}\right)
\end{align*}
and therefore $\E\left[|R_s^{r,\eta;p,q}|^2\right]\leq C_{p, q, r, \eta}$ for some constant $ C_{p, q, r, \eta}$.

The same argument applies to the truncated terms $R_t^{j,k;m,\ell,(L)}$.

With these bounds in place, applying the  Cauchy–Schwarz inequality to 
\eqref{eq:cov_partial} for $|s-t|>2L+\max\{j,k,r,\ell\}$ yields
\[
|\Cov(R_t,\tilde R_s)|
\le C
\left(\sum_{h>L}\|A_h\|\right)\;.
\]
for some constant $C$.
  It follows that $\Cov(R_t, \tilde R_s)$ tends to $0$ as $|s-t|\rightarrow \infty$, since $\sum_{h>L}\|A_h\|\rightarrow 0$ as $L\rightarrow \infty$.

\textit{Proof of }(B) 
Our goal is to prove that
\[
\sum_{t=1}^n \mathbb{E}\left[ M_{t,n}^2(u) \, \mathbf{1}\left(|M_{t,n}(u)| >   \varepsilon\right) \,\middle|\, \mathcal{F}_{t-1} \right] \xrightarrow{\mathbb{P}} 0.
\]
For this, we fix \( \delta > 0 \). Then, by Markov's inequality  and the tower property for conditional expectations, it follows that
\begin{align*}
&\mathbb{P} \left( \sum_{t=1}^n \mathbb{E} \left[ M_{t,n}^2(u) \, \mathbf{1}\left(|M_{t,n}(u)| >   \varepsilon\right) \,\middle|\, \mathcal{F}_{t-1} \right]  > \delta \right) \\
&\leq \frac{1}{\delta}  \sum_{t=1}^n \mathbb{E} \left[ \mathbb{E} \left[ M_{t,n}^2(u) \, \mathbf{1}\left(|M_{t,n}(u)| >   \varepsilon\right) \,\middle|\, \mathcal{F}_{t-1} \right] \right]  \\
&= \frac{1}{\delta}  \sum_{t=1}^n \mathbb{E} \left[ M_{t,n}^2(u) \, \mathbf{1}\left(|M_{t,n}(u)| >   \varepsilon\right) \right] \;. 
\end{align*}
It remains to show that the right hand side vanishes as $n\to \infty.$ For this, we  write
\[
\mathcal{Z}_{t} := \|\mathbf{Z}_{t-1}\| + \sqrt{\E[\|\mathbf{Z}_0\|^2]}.
\]

From the Lipschitz continuity of each $f_m$ and  \eqref{eq:almost_sure_bound_last}
in the proof of Lemma~\ref{lemma:lemma3.4Furmanczyk}, we have for some constant $C>0$ (depending on $M$, the Lipschitz constant of $f_m$ and $\E[\|\mathbf Z_0\|^2]$),
\[
|U_{t+j,j,m}|
\le
C\,\|A_j\|\,
\mathcal{Z}_t, 
\]
and thus
\begin{align*}
|M_{t,n}(u)|
&\le
\frac1{\sqrt n}
\sum_{j=0}^{u-1}\sum_{m=1}^M |a_{t+j,m}^{(n)}|\,|U_{t+j,j,m}|\\
&\le
\frac{C}{\sqrt n}
\sum_{j=0}^{u-1}\sum_{m=1}^M |a_{t+j,m}^{(n)}|\|A_j\|\mathcal Z_t
\le
\frac{C}{\sqrt n}\left(\sum_{j=0}^{u-1}\max_{m\in \{1, \ldots, M\}}|a_{t+j,m}^{(n)}|\right)\mathcal Z_t.
\end{align*}
 By assumption \eqref{eq:W2-multi}
\[
\max_{1\le t\le n}\max_{1\le m\le M} \big|a_{t,m}^{(n)}\big|
=o(\sqrt{n}),
\]
so that, with $\mathcal{C} := \sup_{j\ge 0}\,\|A_j\|\,\max_{1\leq m\leq M}\mathrm{Lip}(f_m)$, 
\[
\gamma_n
:=
\frac{\sqrt n\,\varepsilon}{Cu\max_{1\le t\le n}\max_{1\le m\le M} \big|a_{t,m}^{(n)}\big|}
\xrightarrow[n\to\infty]{}\infty. 
\]

From the bound above we obtain
\[
M_{t,n}^2(u)\mathbf 1(|M_{t,n}(u)|>\varepsilon)
\le
\frac{C}{n}\Big( \sum_{m=1}^M\sum_{j=0}^{u-1} |a_{t+j,m}^{(n)}|\Big)^2
\mathcal Z_t^2
\,\mathbf 1(\mathcal Z_t>\gamma_n).
\]
It follows that
\begin{align*}
\sum_{t=1}^n
\E\big[ M_{t,n}^2(u) \mathbf 1(|M_{t,n}(u)|>\varepsilon)\big]
&\le
\frac{C^2}{n}
\sum_{t=1}^n
\Big( \sum_{m=1}^M\sum_{j=0}^{u-1}|a_{t+j,m}^{(n)}|\Big)^2
\sup_{\ell \in \mathbb{Z}}\E\big[\mathcal Z_\ell^2\mathbf 1(\mathcal Z_\ell>\gamma_n)\big].
\end{align*}
An application of the Cauchy–Schwarz inequality yields
\[
\Big(\sum_{m=1}^M\sum_{j=0}^{u-1}|a_{t+j,m}^{(n)}|\Big)^2
\le
u M \sum_{m=1}^M
\sum_{j=0}^{u-1} (a_{t+j,m}^{(n)})^2\;,
\]
so that by assumption \eqref{eq:W1-multi},
\[
\frac{1}{n}\sum_{t=1}^n\Big(\sum_{m=1}^M\sum_{j=0}^{u-1}|a_{t+j,m}^{(n)}|\Big)^2
\le
C_{u,M} \sum_{m=1}^M \left( \frac{1}{n}\sum_{t=1}^{n}(a_{t,m}^{(n)})^2\right)
=\mathcal{O}(1)\;.
\] 
Moreover $(\mathcal Z_\ell)_{\ell\in\mathbb{Z}}$ has uniformly bounded second moments and is uniformly integrable (by the moment assumptions on $\mathbf Z_0$), so that
\[
\sup_{\ell\in \mathbb{Z}}\E\big[\mathcal Z_\ell^2\mathbf 1(\mathcal Z_\ell>\gamma_n)\big]
\xrightarrow[n\to\infty]{}0.
\]
\paragraph{Proof of $(ii)$}
By \eqref{eq:sigmas-multi},
\[
\sigma_u^2
=
\sum_{j,k=0}^{u-1}\sum_{m,\ell=1}^M
\E[U_{1+j,j,m}U_{1+k,k,\ell}]\,
G_{m\ell}(k-j).
\]
Letting $u\to\infty$, we obtain
\[
\sigma^2
:=
\lim_{u\to\infty}\sigma_u^2
=
\sum_{j,k=0}^{\infty}\sum_{m,\ell=1}^M
\E[U_{1+j,j,m}U_{1+k,k,\ell}]\,
G_{m\ell}(k-j)\;.
\]
The limit $\sigma^2$ exists and it is well defined as the series $\sum_{j,k=0}^{\infty}\sum_{m,\ell=1}^M
\E[U_{1+j,j,m}U_{1+k,k,\ell}]\,
G_{m\ell}(k-j)$ converges. To prove this, recall that if a series of absolute values converges, then the original series also converges. Since the partial sums of the absolute values form a non-decreasing sequence, the existence of any finite upper bound ensures their convergence, and thus the convergence of the original series. By the Cauchy-Schwarz inequality,
\[
\bigl|G_{n,m\ell}(h)\bigr|
=\Bigl|\frac1n\sum_{t=1}^n a^{(n)}_{t,m}\,a^{(n)}_{t+h,\ell}\Bigr|
\le \Bigl(\frac1n\sum_{t=1}^n (a^{(n)}_{t,m})^2\Bigr)^{1/2}
     \Bigl(\frac1n\sum_{t=1}^n (a^{(n)}_{t+h,\ell})^2\Bigr)^{1/2}.
\]
By the fact that $a^{(n)}_{s,\ell}=0$ for $s\notin\{1,\dots,n\}$,
\[
\frac1n\sum_{t=1}^n (a^{(n)}_{t,m})^2 = G_{n,mm}(0),
\qquad
\frac1n\sum_{t=1}^n (a^{(n)}_{t+h,\ell})^2 = G_{n,\ell\ell}(0),
\]
hence $\ |G_{n,m\ell}(h)|\le \sqrt{G_{n,mm}(0)}\,\sqrt{G_{n,\ell\ell}(0)}$.
Letting $n\to\infty$ and using \eqref{eq:W1-multi} for $(m,\ell,h)$, $(m,m,0)$ and $(\ell,\ell,0)$ yields
\[
|G_{m\ell}(h)|\le \sqrt{G_{mm}(0)}\,\sqrt{G_{\ell\ell}(0)}.
\]
Furthermore by Lemma \ref{lemma:lemma3.4Furmanczyk}
\[
|\E[U_{1+j,j,m}U_{1+k,k,\ell}]|
\le
C\,\|A_j\|\|A_k\|\;.
\] 

Hence, given that $\sum_j\|A_j\|<\infty$ we obtain by Assumption~\ref{assum1:M1}, 
\begin{align*}
\sum_{j,k=0}^\infty\sum_{m,\ell=1}^M
\big|\E[U_{1+j,j,m}U_{1+k,k,\ell}]\,G_{m\ell}(k-j)\big|
&\le
C \sqrt{G_{\ell, \ell}(0})\sqrt{G_{m,m}(0})
\sum_{j,k=0}^\infty \|A_j\|\|A_k\| \\
&=
C \Big(\sum_{j=0}^\infty\|A_j\|\Big)^2 <\infty,
\end{align*}
Thus $\sigma^2$ is well-defined and finite. Moreover, $\sigma^2$ can be rewritten as
$$\sigma^2=\sum_{j,k=0}^{\infty}\sum_{m,\ell=1}^M
\E[U_{1+j,j,m}U_{1+k,k,\ell}]\,
G_{m\ell}(k-j)\
= \sum_{w\in \mathbb{Z}} \sum_{m,\ell=1}^MG_{m\ell}(w)\sum_{s,p=0}^\infty \mathbb{E}[ U_{1,s,m}  U_{1+w,p,\ell}] \;.
$$
The reindexing above only reformulates the double sum and most terms in the inner $\sum_{s,p}$ vanish by martingale orthogonality, so effectively only the pairs $(s,p)$ that correspond to the same innovation contribute (specifically this is the content of Lemma 3.5 in  \cite{Furmanczyk}, first equation after (5.2) of its proof). 
Since
$$
U_{t,j,m}
:=\E\big[f_m(\mathbf{X}_t)\mid\mathcal{F}_{t-j}\big]-\E\big[f_m(\mathbf{X}_t)\mid\mathcal{F}_{t-j-1}\big],
$$
and by the martingale decomposition
$
f_m(\mathbf{X}_t)-\E [f_m(\mathbf{X}_t)]=\sum_{j=0}^\infty U_{t,j,m}$ in $L^2(\Omega, \mathcal{F}, \mathbb{P}),$ we have
$$
\sum_{s,p=0}^\infty\E[U_{1,s,m}\,U_{1+w,p,\ell}]
=\E\big[(f_m(\mathbf{X}_1)-\E [f_m(\mathbf{X}_1)])(f_\ell(\mathbf{X}_{1+w})-\E [f_\ell(\mathbf{X}_{1+w})])\big].
$$
Putting all pieces together,
\[
\sigma^2
=
\sum_{w\in\mathbb{Z}}
\sum_{m,\ell=1}^M
G_{m\ell}(w)\,\Cov\big(f_m(\mathbf X_1),f_\ell(\mathbf X_{1+w})\big),
\]
which coincides with the variance formula stated in Theorem~\ref{thm:main}.

\medskip
\paragraph{Proof of $(iii)$}
We finally show that the truncation error vanishes in probability. By Chebyshev's inequality,
\begin{align*}
\P\Big(|V_{u,n}-n^{-1/2}S_n|\ge\varepsilon\Big)
&\le
\frac{1}{\varepsilon^2 n}
\E\Big[
\Big(\sum_{t=1}^n\sum_{j=u}^\infty W_{t,j}^{(n)}\Big)^2
\Big]\\
&=
\frac{1}{\varepsilon^2 n}
\sum_{t,s=1}^n\sum_{i,j=u}^\infty
\E\big[ W_{t,j}^{(n)} W_{s,i}^{(n)}\big].
\end{align*}
Expanding $W_{t,j}^{(n)}=\sum_m a_{t,m}^{(n)}U_{t,j}^{(m)}$ and using stationarity and orthogonality of the $U_{t,j}^{(m)}$ (Lemma~6.4 of \cite{ho1997limit} applied componentwise),
\[
\E[W_{t,j}^{(n)}W_{s,i}^{(n)}]
=
\sum_{m,\ell=1}^M a_{t,m}^{(n)}a_{s,\ell}^{(n)}
\E\big[U_{t,j}^{(m)}U_{s,i}^{(\ell)}\big]
=
\sum_{m,\ell=1}^M a_{t,m}^{(n)}a_{t-j+i,\ell}^{(n)}
\E\big[U_{t,j}^{(m)}U_{t-j+i,i}^{(\ell)}\big],
\]
so that
\begin{align*}
\P\Big(|V_{u,n}-n^{-1/2}S_n|\ge\varepsilon\Big)
&\le
\frac{1}{\varepsilon^2 n}
\sum_{t=1}^n\sum_{i,j=u}^\infty
\sum_{m,\ell=1}^M
\big|a_{t,m}^{(n)}a_{t-j+i,\ell}^{(n)}\big|\,
\E\big|U_{t,j}^{(m)}U_{t-j+i,i}^{(\ell)}\big|.
\end{align*}
By the Cauchy–Schwarz inequality and Lemma~3.4 of \cite{Furmanczyk},
\[
\E\big|U_{t,j}^{(m)}U_{t-j+i,i}^{(\ell)}\big|
\le
\sqrt{\E\big[ (U_{t,j}^{(m)})^2\big]\,\E\big[(U_{t-j+i,i}^{(\ell)})^2\big]}
\le
C\,\|A_j\|\|A_i\|,
\]
for some constant $C$ independent of $t,i,j,m,\ell$. Thus
\begin{align*}
\P\Big(|V_{u,n}-n^{-1/2}S_n|\ge\varepsilon\Big)
&\le
\frac{C}{\varepsilon^2 n}
\sum_{i,j=u}^\infty \|A_i\|\|A_j\|
\sum_{t=1}^n
\sum_{m,\ell=1}^M
\big|a_{t,m}^{(n)}a_{t+i-j,\ell}^{(n)}\big|.
\end{align*}
Using the Cauchy–Schwarz inequality and \eqref{eq:W1-multi},
\begin{align*}
\frac{1}{n}\sum_{t=1}^n
\sum_{m,\ell}|a_{t,m}^{(n)}a_{t+h,\ell}^{(n)}|
\le&
C_M\Big(\frac{1}{n}\sum_{t=1}^n\sum_m (a_{t,m}^{(n)})^2\Big)^{1/2}
\Big(\frac{1}{n}\sum_{t=1}^n\sum_\ell(a_{t+h,\ell}^{(n)})^2\Big)^{1/2}\\
\le&
C_M \sqrt{G_{m,m}(0)}\sqrt{G_{\ell,\ell}(0)},
\end{align*}
so that
\[
\limsup_{n\to\infty}
\P\Big(|V_{u,n}-n^{-1/2}S_n|\ge\varepsilon\Big)
\le
\frac{C}{\varepsilon^2}\Big(\sum_{j\ge u}\|A_j\|\Big)^2,
\]
for some constant $C$. Since $\sum_{j\ge 0}\|A_j\|<\infty$, the right–hand side tends to 0 as $u\to\infty$.

\section{Technical Lemmas}
\label{appendix:Lemmas}

    The next result shows that the conclusion of Lemma 3.4 of  \cite{Furmanczyk} still holds under the assumption of the previous Theorem. 
\begin{lemma}[Adapted Lemma~3.4 of \cite{Furmanczyk}]
\label{lemma:lemma3.4Furmanczyk}
Suppose $(\mathbf{X}_t)_{t\geq 1}$ satisfies Assumption \ref{assum1:M1} and the function $f$ is Lipschitz with constant $L$.   
Then, for all $m$ and all $t \geq 1$, the variables $U_{t,j }$ satisfy:
\begin{enumerate}
    \item[(i)] $\mathbb{E}\!\left[\, U_{t,0 }^2\,\right] \leq \mathbb{E}\!\left[     f(\mathbf{X}_t)^2\right]$.
    \item[(ii)] $\mathbb{E}\, [\, U_{t,j }^2\,] \leq \mathcal{C} \,\|A_j\|  ^2$,  for all $j \geq 1$, where the constant $\mathcal{C}=\mathcal{C}(L)$ is independent of $t,j $.
\end{enumerate}
\end{lemma}

\begin{proof}
Define, for all $s\geq 1$
\[
f_s(\mathbf{x}) := \E\big[f(\mathbf{X}_{t,s-1}+\mathbf{x})\big] = \int_{\mathbb{R}^k} f(\mathbf{z}+\mathbf{x})\, dF_{s-1}(\mathbf{z}),
\]
where $\mathbf{X}_{t,s}=\sum_{i=0}^s A_i\mathbf{Z}_{t-i}$ is the truncated process and $F_{s-1}$ is the law of $\mathbf{X}_{t,s-1}$ (which is independent of $t$ because of stationarity of $(\mathbf{Z}_j)_j$). By the martingale decomposition, $
U_{t,j } =     f_j(\mathbf{R}_{t,j}) -     f_{j+1}(\mathbf{R}_{t,j+1}),$ 
where $\mathbf{R}_{t,j} = \sum_{i=j+1}^\infty A_i\mathbf{Z}_{t-i}$ and $
\mathbf{R}_{t,j} - \mathbf{R}_{t,j+1} = A_{j+1}\mathbf{Z}_{t-j-1}\;.$ Since $f$ is Lipschitz,
\[
|f_s(\mathbf x)-f_s(\mathbf y)|
=\Big|\,\E\big[f(\mathbf X_{t,s-1}+\mathbf x)-f(\mathbf X_{t,s-1}+\mathbf y)\big]\,\Big|
\le L\,\|\mathbf x-\mathbf y\|_{\mathcal H},
\]
and so $\sup_{s\ge 1}\mathrm{Lip}(f_s)\le L<\infty$.

Next, let $F_{1,j}$ denote the Borel probability measure on $\mathbb{R}^k$ corresponding to the law of $A_{j+1}\mathbf{Z}_1$.  
By the independence of $A_{j+1}\mathbf{Z}_{t-j-1}$ and $\mathbf{R}_{t,j+1}$, 
\[
    f_{j+1}(\mathbf{R}_{t,j+1}) 
= \int_{\mathbb{R}^k}     f_j(\mathbf{R}_{t,j+1} - t) \, \,dF_{1,j}(t).
\]
Therefore,
\begin{align}
\label{eq:almost_sure_bound}
|U_{t,j }|
&\leq \int_{\mathbb{R}^k} \big|     f_j(\mathbf{R}_{t,j}) -     f_j(\mathbf{R}_{t,j+1} - t) \big| \, dF_{1,j}(t) \\
&\leq \mathrm{Lip}(    f_j) \int_{\mathbb{R}^k} \| \mathbf{R}_{t,j} - (\mathbf{R}_{t,j+1} - t) \|_{\mathbb{R}^k} \, dF_{1,j}(t) \\
&= \mathcal{C} \int_{\mathbb{R}^k} \| A_{j+1}\mathbf{Z}_{t-j-1} - t \|   \,d F_{1,j}(t) \\
&\leq \mathcal{C} \,\|A_{j+1}\|\left( \|\mathbf{Z}_{t-j-1}\|   + \sqrt{\mathbb{E}\|\mathbf{Z}_0\|  ^2} \right).
\label{eq:almost_sure_bound_last}
\end{align}
Squaring, applying $(a+b)^2 \leq 2a^2 + 2b^2$, and taking expectations yields
\[
\mathbb{E}[\,U_{t,j }^2\,] \leq \mathcal{C}' \, \|A_{j+1}\|  ^2,
\]
where $\mathcal{C}'$ depends only on $\mathcal{C}$ and $\mathbb{E}\|\mathbf{Z}_0\|  ^2$.

For $j=0$,  
\[
U_{t,0 } =     f(\mathbf{X}_t) - \mathbb{E}[\,    f(\mathbf{X}_t) \mid \mathcal{F}_{t-1}\,].
\]
By Jensen's inequality applied to the conditional expectation,
\[
\mathbb{E}[U_{t,0 }^2] 
\leq \mathbb{E}[\,    f(\mathbf{X}_t)^2\,].
\]

This proves (i) and (ii).
\end{proof}

\begin{lemma}\label{lem:BnM1-multi}
For any $j,k\ge 0$ and any $m,\ell\in\{1,\dots,M\}$,
\begin{align*}
\bigg|
\sum_{t=1}^n a^{(n)}_{t+j,m}\,a^{(n)}_{t+k,\ell}
-
\sum_{s=1}^n a^{(n)}_{s,m}\,a^{(n)}_{s+k-j,\ell}
\bigg|
\;\le\;
2j\,
\max_{t,m}|a^{(n)}_{t,m}|^2 .
\end{align*}
\end{lemma}
\begin{proof}
Let $s:=t+j$. Splitting the sum on the left as
$\sum_{t=1}^n=\sum_{t=1}^{n-j}+\sum_{t=n-j+1}^n$
and the sum on the right as
$\sum_{s=j+1}^n=\sum_{s=1}^n-\sum_{s=1}^j$,
we obtain
\begin{align*}
\sum_{t=1}^n a^{(n)}_{t+j,m}\,a^{(n)}_{t+k,\ell}
&=
\sum_{s=1}^{n} a^{(n)}_{s,m}\,a^{(n)}_{s+k-j,\ell}
\\
&\quad
-
\sum_{t=n-j+1}^n a^{(n)}_{t+j,m}\,a^{(n)}_{t+k,\ell}
-
\sum_{s=1}^j a^{(n)}_{s,m}\,a^{(n)}_{s+k-j,\ell}.
\end{align*}
Subtracting $\sum_{s=1}^{n} a^{(n)}_{s,m}\,a^{(n)}_{s+k-j,\ell}$ from both sides, taking absolute values, and bounding each coefficient by
$\max_{t,m}|a^{(n)}_{t,m}|$, we obtain
\[
\bigg|
\sum_{t=1}^n a^{(n)}_{t+j,m}\,a^{(n)}_{t+k,\ell}
-
\sum_{s=1}^n a^{(n)}_{s,m}\,a^{(n)}_{s+k-j,\ell}
\bigg|
\le
2j\,\max_{t,m}|a^{(n)}_{t,m}|^2,
\]
which concludes the proof.
\end{proof}

\begin{lemma}
\label{lem:residual_to_zero}
Under the assumptions of Theorem~\ref{thm:main}, for each fixed $u\in\mathbb{N}$,
the remainder term $H_{u,n}$ defined in \eqref{eq:H_un-multi} satisfies
\[
H_{u,n}\xrightarrow{\P}0\qquad (n\to\infty).
\]
\end{lemma}

\begin{proof}
Recall that
\[
H_{u,n}
=
\frac{1}{\sqrt n}\sum_{j=0}^{u-1}
\bigg(
\sum_{t=1}^j W_{t,j}^{(n)}
-
\sum_{t=n+1}^{n+j} W_{t,j}^{(n)}
\bigg),
\qquad
W_{t,j}^{(n)}=\sum_{m=1}^M a_{t,m}^{(n)}U_{t,j,m}.
\]
By the standing assumption $a_{t,m}^{(n)}=0$ for all $t\ge n+1$, the second boundary term vanishes, hence
\[
H_{u,n}
=
\frac{1}{\sqrt n}\sum_{j=0}^{u-1}\sum_{t=1}^j W_{t,j}^{(n)}.
\]

Fix $j\in\{0,\dots,u-1\}$. Using Lemma~6.4 in \cite{ho1997limit}, we have for $s\neq t$
\[
\E\!\big[U_{t,j,m}\,U_{s,j,\ell}\big]=0,
\qquad\text{since }(t-j)\neq(s-j),
\]
and therefore
\begin{align*}
\Var\!\Big(\sum_{t=1}^j W_{t,j}^{(n)}\Big)
&=
\sum_{t=1}^j \E\!\big[(W_{t,j}^{(n)})^2\big].
\end{align*}
Moreover, by the Cauchy-Schwarz inequality and the inequality $(\sum_{m=1}^M b_m)^2\le M\sum_{m=1}^M b_m^2$,

\begin{align}
\E\!\big[(W_{t,j}^{(n)})^2\big]
&=\E\!\Big[\Big(\sum_{m=1}^M a_{t,m}^{(n)}U_{t,j,m}\Big)^2\Big]
\\
&=
\sum_{m,\ell=1}^M
a_{t,m}^{(n)} a_{t,\ell}^{(n)}
\E\!\big[U_{t,j,m}U_{t,j,\ell}\big]
\nonumber\\
&\le
\sum_{m,\ell=1}^M
|a_{t,m}^{(n)}|\,|a_{t,\ell}^{(n)}|
\big|\E\!\big[U_{t,j,m}U_{t,j,\ell}\big]\big|
\nonumber\\
&\le
\sum_{m,\ell=1}^M
|a_{t,m}^{(n)}|\,|a_{t,\ell}^{(n)}|
\sqrt{\E[U_{t,j,m}^2]}\sqrt{\E[U_{t,j,\ell}^2]}
\nonumber\\
&=
\Big(\sum_{m=1}^M |a_{t,m}^{(n)}|\sqrt{\E[U_{t,j,m}^2]}\Big)^2\leq
M\sum_{m=1}^M (a_{t,m}^{(n)})^2\,\E[U_{t,j,m}^2].
\end{align}

By Lemma~3.4 in \cite{Furmanczyk} applied to each Lipschitz $f_m$ (and using that $M$ is fixed and the Lipschitz constants are uniformly bounded by assumption), there exists a constant $C>0$ such that for all $m$ and all $t$,
\[
\E[U_{t,j,m}^2]\le C\|A_j\|^2.
\]
Hence,
\[
\Var\!\Big(\sum_{t=1}^j W_{t,j}^{(n)}\Big)
\le
CM\|A_j\|^2 \sum_{t=1}^j\sum_{m=1}^M (a_{t,m}^{(n)})^2.
\]

By Chebyshev's inequality and the bound
$\Var(\sum_{j=0}^{u-1}Y_j)\le u\sum_{j=0}^{u-1}\Var(Y_j)$,
for any $\varepsilon>0$,
\begin{align*}
\P\Big(|H_{u,n}|>\varepsilon\Big)
&=
\P\Big(\Big|\sum_{j=0}^{u-1}\sum_{t=1}^j W_{t,j}^{(n)}\Big|>\sqrt n\,\varepsilon\Big)
\\
&\le
\frac{1}{n\varepsilon^2}\,
\Var\Big(\sum_{j=0}^{u-1}\sum_{t=1}^j W_{t,j}^{(n)}\Big)
\\
&\le
\frac{u}{n\varepsilon^2}\sum_{j=0}^{u-1}
\Var\Big(\sum_{t=1}^j W_{t,j}^{(n)}\Big)
\\
&\le
\frac{CuM}{n\varepsilon^2}\sum_{j=0}^{u-1}\|A_j\|^2
\sum_{t=1}^j\sum_{m=1}^M (a_{t,m}^{(n)})^2.
\end{align*}
Since $u$ is fixed, the sums in $j$ and $t$ are over finitely many indices. By \eqref{eq:W2-multi},
$\max_{t\le u,\;m\le M}|a_{t,m}^{(n)}|=o(\sqrt n)$, hence
\[
\sum_{j=0}^{u-1}\sum_{t=1}^j\sum_{m=1}^M (a_{t,m}^{(n)})^2
\;\le\;
u^2M\max_{t\le u,\;m\le M}(a_{t,m}^{(n)})^2
\;=\;o(n).
\]
Therefore the right-hand side is $o(1)$, and we conclude that
$\P(|H_{u,n}|>\varepsilon)\to 0$ for every $\varepsilon>0$. 
\end{proof}

\section{Proofs for Section \ref{sec:fcns}}
\label{appendix:missing_proofs}

\begin{proof}[Proof of Lemma \ref{lem:fcn_res_representation}]
We argue by induction on \(\ell\). For $\ell=1$, since \(m_0=1\), we have
\[
\mathbf H_t^{(1)}
=
I^{(1)}x_t+
\sigma\!\Big(
\mathbf b^{(1)}+\sum_{j=0}^{k_1-1}W_j^{(1)}x_{t+j}
\Big),
\]
so \eqref{eq:induction_claim_fcn} holds with \(K_1=k_1\) and
\[
\mathbf g_1(u_1,\ldots,u_{k_1})
:=
I^{(1)}u_1+
\sigma\!\Big(
\mathbf b^{(1)}+\sum_{j=0}^{k_1-1}W_j^{(1)}u_{j+1}
\Big).
\]
This map is Lipschitz, since it is the sum of a linear map and a Lipschitz map.

For the induction step $\ell-1\Longrightarrow \ell$, fix \(2\le \ell\le L\), and assume that for layer \(\ell-1\) there exists a deterministic Lipschitz function
\(
\mathbf g_{\ell-1}:\mathbb{R}^{K_{\ell-1}}\to\mathbb{R}^{m_{\ell-1}}
\)
such that
\[
\mathbf H_t^{(\ell-1)}
=
\mathbf g_{\ell-1}\bigl(x_t,\ldots,x_{t+K_{\ell-1}-1}\bigr)
\qquad\text{for all }t\ge 1.
\]
Then, by definition of the network,
\[
\mathbf H_t^{(\ell)}
=
I^{(\ell)}\mathbf H_t^{(\ell-1)}
+
\sigma\!\Big(
\mathbf b^{(\ell)}+\sum_{j=0}^{k_\ell-1}W_j^{(\ell)}\mathbf H_{t+j}^{(\ell-1)}
\Big).
\]
Using the induction hypothesis,
\[
\mathbf H_{t+j}^{(\ell-1)}
=
\mathbf g_{\ell-1}\bigl(x_{t+j},\ldots,x_{t+j+K_{\ell-1}-1}\bigr),
\qquad j=0,\ldots,k_\ell-1.
\]
Hence \(\mathbf H_t^{(\ell)}\) depends only on
\(
x_t,\ldots,x_{t+k_\ell+K_{\ell-1}-2}.
\)
Since
\(
k_\ell+K_{\ell-1}-1
=
\sum_{q=1}^{\ell}k_q-(\ell-1)
=
K_\ell,
\)
this is exactly the block \((x_t,\ldots,x_{t+K_\ell-1})\). Therefore \eqref{eq:induction_claim_fcn} holds with
\[
\mathbf g_\ell(u_1,\ldots,u_{K_\ell})
:=
I^{(\ell)}\mathbf g_{\ell-1}(u_1,\ldots,u_{K_{\ell-1}})
+
\sigma\!\Big(
\mathbf b^{(\ell)}
+
\sum_{j=0}^{k_\ell-1}
W_j^{(\ell)}
\mathbf g_{\ell-1}(u_{j+1},\ldots,u_{j+K_{\ell-1}})
\Big).
\]
This map is deterministic and Lipschitz, because \(\mathbf g_{\ell-1}\) is Lipschitz, linear maps preserve Lipschitz continuity, and \(\sigma\) is Lipschitz componentwise.
\end{proof}

\begin{proof}[Proof of Theorem \ref{thm:fcn_gap}]

By Lemma~\ref{lem:fcn_res_representation}, there exists a function
$\mathbf g_L:\mathbb{R}^{K_L}\to\mathbb{R}^{m_L}$ (independent of $t$) such that
\begin{equation}\label{eq:H_equals_gX}
\mathbf H^{(L)}_t=\mathbf g_L(\mathbf X_t),\qquad t\ge 1.
\end{equation}
Write $\mathbf g_L=(g_1,\ldots,g_{m_L})^\top$ and define
\[
 \mathbb{E}\big[\mathbf H^{(L)}_1\big]
= \mathbb{E}\big[\mathbf g_L(\mathbf X_1)\big]
= \big(\mathbb{E}[g_1(\mathbf X_1)],\ldots,\mathbb{E}[g_{m_L}(\mathbf X_1)]\big)^\top.
\]
Then
\[
\sqrt n\Big(\mathrm{GAP}(\mathbf x)- \mathbb{E}\big[\mathbf H^{(L)}_1\big]\Big)
=
\frac{1}{\sqrt n}\sum_{t=1}^n\Big(\mathbf g_L(\mathbf X_t)-\mathbb{E}[\mathbf g_L(\mathbf X_t)]\Big).
\]
We now apply Theorem~\ref{thm:main} with $M=m_L$, $f_m=g_m$, and deterministic weights
\[
a_{t,m}^{(n)}:=1\cdot \mathbf 1_{\{1\le t\le n\}},
\qquad t\ge 1,\ m=1,\ldots,m_L.
\]
Clearly $a_{t,m}^{(n)}=0$ for $t\ge n+1$. Condition \eqref{eq:W2-multi} holds since
\(
\max_{1\le t\le n} \, \max_{1\le m\le m_L} \, |a_{t,m}^{(n)}|=1=o(\sqrt n).
\)
For \textup{(W1)}, fix $m,\ell\in\{1,\ldots,m_L\}$ and $h\in\mathbb{Z}$. Using the convention
$a_{t+h,\ell}^{(n)}=0$ if $t+h\notin\{1,\ldots,n\}$,
\[
G_{n,m\ell}(h)
=
\frac{1}{n}\sum_{t=1}^n a_{t,m}^{(n)}a_{t+h,\ell}^{(n)}
=
\frac{1}{n}\#\{t\in\{1,\ldots,n\}:1\le t+h\le n\}
=
\begin{cases}
\frac{n-|h|}{n}, & |h|<n,\\[2mm]
0, & |h|\ge n,
\end{cases}
\]
hence for every fixed $h$,
\(
G_{n,m\ell}(h)\xrightarrow[n\to\infty]{}1 \;=:\; G_{m\ell}(h).
\)
Therefore Theorem~\ref{thm:main} yields
\[
\frac{1}{\sqrt n}\sum_{t=1}^n\Big(\mathbf g_L(\mathbf X_t)-\mathbb{E}[\mathbf g_L(\mathbf X_t)]\Big)
\xrightarrow{\mathcal D}\mathcal N(0,\Sigma^{\mathrm{GAP}}),
\]
where, since $G_{m\ell}(h)\equiv 1$ for all $h\in\mathbb{Z}$, the covariance entries are
\[
\Sigma^{\mathrm{GAP}}_{m\ell}
=
\Cov\!\big(g_m(\mathbf X_1),g_\ell(\mathbf X_1)\big)
+
2\sum_{h=1}^\infty
\Cov\!\big(g_m(\mathbf X_1),g_\ell(\mathbf X_{1+h})\big)\;.
\]
\end{proof}

\begin{lemma}
\label{lem:cnn_trivial_case}
Let \(W^{(1)}_0,\ldots,W^{(1)}_{k_1-1}\in\mathbb R^{m_1\times 1}\) be filters, and let \((X_t)_{t\in\mathbb Z}\) be the Gaussian AR\((1)\) process
\[
X_t=\theta X_{t-1}+\varepsilon_t,\qquad \theta\in(-1,1),
\]
where \((\varepsilon_t)_{t\in\mathbb Z}\) is i.i.d.\ with \(\varepsilon_t\sim\mathcal N(0,\sigma_\varepsilon^2)\). Consider the one-layer network
\[
\mathrm{GAP}(\mathbf x)
=
\frac1n\sum_{t=1}^n \mathrm{ReLU}\!\left(\sum_{j=0}^{k_1-1}W_j^{(1)}x_{t+j}\right)
=
\frac1n\sum_{t=1}^n \mathbf g(\mathbf X_t),
\]
where \(\mathbf X_t=(X_t,\ldots,X_{t+k_1-1})^\top\), \(\mathbf W\in\mathbb R^{m_1\times k_1}\) is the matrix with columns \(W_0^{(1)},\ldots,W_{k_1-1}^{(1)}\), and $
\mathbf g(\cdot)=\mathrm{ReLU}(\mathbf W\cdot).$
Let \(\Sigma_{j,\ell}^{\mathrm{GAP}}(\theta)\) be the \((j,\ell)\)-th entry of the long-run covariance matrix \eqref{eq:long_run_cov}, and let
\[
\mathcal R_{j,\ell}^{\mathrm{GAP}}(\theta)
=
\frac{\Sigma_{j,\ell}^{\mathrm{GAP}}(\theta)}
{\sqrt{\Sigma_{j,j}^{\mathrm{GAP}}(\theta)\Sigma_{\ell,\ell}^{\mathrm{GAP}}(\theta)}}.
\]
Define, for \(h\ge 0\),
\[
\rho_{j,\ell}(h,\theta)
:=
\mathrm{Corr}\!\bigl((\mathbf W\mathbf X_1)_j,(\mathbf W\mathbf X_{1+h})_\ell\bigr),
\]
and
\[
\Phi_{j,\ell}(h;\theta)
:=
\sqrt{1-\rho_{j,\ell}(h,\theta)^2}
+
\rho_{j,\ell}(h,\theta)\bigl(\pi-\arccos(\rho_{j,\ell}(h,\theta))\bigr)
-1.
\]
Then, for all \(j,\ell\in\{1,\ldots,m_1\}\),
\[
\mathcal R_{j,\ell}^{\mathrm{GAP}}(\theta)
=
\frac{\Phi_{j,\ell}(0;\theta)+2\sum_{h=1}^\infty \Phi_{j,\ell}(h;\theta)}
{\sqrt{\Phi_{j,j}(0;\theta)+2\sum_{h=1}^\infty \Phi_{j,j}(h;\theta)}
 \sqrt{\Phi_{\ell,\ell}(0;\theta)+2\sum_{h=1}^\infty \Phi_{\ell,\ell}(h;\theta)}}.
\]

Moreover, if \(k_1=1\), then for all \(j,\ell\in\{1,\ldots,m_1\}\),
\[
\mathcal R_{j,\ell}^{\mathrm{GAP}}(0)
=
\begin{cases}
1, & \text{if } W^{(1)}_{0,j}W^{(1)}_{0,\ell}\geq 0,\\[1.2ex]
-\dfrac{1}{\pi-1}, & \text{if } W^{(1)}_{0,j}W^{(1)}_{0,\ell}<0,
\end{cases}
\]
and
\[
\lim_{\theta\to1^-}\mathcal R_{j,\ell}^{\mathrm{GAP}}(\theta)
=
\begin{cases}
1, & \text{if } W^{(1)}_{0,j}W^{(1)}_{0,\ell} \geq 0,\\[1.2ex]
\displaystyle \frac{\log 2-2}{\pi+\log 2-2}, & \text{if } W^{(1)}_{0,j}W^{(1)}_{0,\ell}<0.
\end{cases}
\]
\end{lemma}
\begin{proof}
Since  $(X_t)_{t\in \mathbb{Z}}$ is a Gaussian AR$(1)$ process,
\begin{align*}
\Cov(X_t,X_{t+u})
=
\frac{\sigma_\varepsilon^2}{1-\theta^2}\theta^{|u|},
\qquad u\in\mathbb{Z}.
\end{align*} 
Denoting by $(\mathbf{W}\mathbf{X}_t)_i$ the $i$-th coordinate of $\mathbf{W}\mathbf{X}_t$ and using bilinearity of the covariance gives
\begin{align}
\Var((\mathbf{W}\mathbf{X}_1)_{j})
&=
\frac{\sigma_\varepsilon^2}{1-\theta^2}
\sum_{r=0}^{k_1-1}\sum_{q=0}^{k_1-1}
W^{(1)}_{r,j}W^{(1)}_{q,j}\,\theta^{|q-r|},
\\
\Var((\mathbf{W}\mathbf{X}_{1+h})_\ell)
&=
\frac{\sigma_\varepsilon^2}{1-\theta^2}
\sum_{s=0}^{k_1-1}\sum_{u=0}^{k_1-1}
W^{(1)}_{s,\ell}W^{(1)}_{u,\ell}\,\theta^{|u-s|},
\\
\Cov\left((\mathbf{W}\mathbf{X}_{1})_j),(\mathbf{W}\mathbf{X}_{1+h})_\ell)\right)
&=
\frac{\sigma_\varepsilon^2}{1-\theta^2}
\sum_{r=0}^{k_1-1}\sum_{s=0}^{k_1-1}
W^{(1)}_{r,j}W^{(1)}_{s,\ell}\,\theta^{|h+s-r|}. \label{eq:cross_relu_cov}
\end{align}
Since the input process is Gaussian, the pair $\left( \mathbf{W}\mathbf{X}_{1},  \mathbf{W}\mathbf{X}_{1+h}\right)$
is jointly Gaussian. Thus, 
\begin{align*}
\left(\mathbf{W}\mathbf{X}_1\right)_j=\sqrt{\Var(\left(\mathbf{W}\mathbf{X}_1\right)_j)}\,U,
\qquad
\left(\mathbf{W}\mathbf{X}_{1+h}\right)_\ell=\sqrt{\Var(\left(\mathbf{W}\mathbf{X}_{1+h}\right)_\ell)}\,V,
\end{align*}
where $(U,V)$ is centered Gaussian with
\begin{align*}
\Var(U)=\Var(V)=1,
\qquad
\Cov(U,V)=\rho_{j,\ell}(h,\theta).
\end{align*}
Since $\mathrm{ReLU}(x)=\max\{0,x\}$, we have the identity 
\begin{align*}
   \Cov\bigl(g_j(\mathbf{X}_1),\;g_\ell(\mathbf{X}_{1+h})\bigr)=\sqrt{\Var(\left(\mathbf{W}\mathbf{X}_1\right)_j) \Var(\left(\mathbf{W}\mathbf{X}_{1+h}\right)_\ell)} \Cov(\mathrm{ReLU}(U),\mathrm{ReLU}(V))\;.
\end{align*}
It remains to derive an expression for \(\Cov(\mathrm{ReLU}(U),\mathrm{ReLU}(V))\). For \((U,V)^\top\)  jointly Gaussian with correlation \(\rho\), Lemma \ref{lem:cov_relu} shows that
\[
\Cov(\mathrm{ReLU}(U),\mathrm{ReLU}(V))
=
\frac{1}{2\pi}
\left(
\sqrt{1-\rho^2}
+
\rho\bigl(\pi-\arccos(\rho)\bigr)
-1
\right).
\]

Next, we set \(\rho=\rho_{j,\ell}(h,\theta)\), and 
\begin{align*}
&\Phi(x):=\sqrt{1-x^2}+x(\pi-\arccos x)-1.\\
&\Phi_{j,\ell}(h;\theta)
:=\Phi(\rho_{j,\ell}(h,\theta))\;.
\end{align*}
We obtain
\begin{align*}
\Cov\bigl(g_j(\mathbf{X}_1),g_\ell(\mathbf{X}_{1+h})\bigr)
=
\frac{\sqrt{\Var((\mathbf{W}\mathbf{X}_1)_j)\Var((\mathbf{W}\mathbf{X}_{1+h})_\ell)}}{2\pi}
\,\Phi_{j,\ell}(h;\theta).
\end{align*}
Substituting the explicit expressions for the variances yields
\begin{align}
\label{eq:explicit_Sigma}
\Sigma^{\mathrm{GAP}}_{j,\ell}(\theta)
&=
\Cov\!\bigl(g_j(\mathbf X_1),g_\ell(\mathbf X_1)\bigr)
+
2\sum_{h=1}^\infty
\Cov\!\bigl(g_j(\mathbf X_1),g_\ell(\mathbf X_{1+h})\bigr)
\nonumber\\
&=
\frac{\sigma_\varepsilon^2}{2\pi(1-\theta^2)}
\Bigg[
\left(
\sum_{r=0}^{k_1-1}\sum_{q=0}^{k_1-1}
W^{(1)}_{r,j}W^{(1)}_{q,j}\,\theta^{|q-r|}
\right)
\left(
\sum_{s=0}^{k_1-1}\sum_{u=0}^{k_1-1}
W^{(1)}_{s,\ell}W^{(1)}_{u,\ell}\,\theta^{|u-s|}
\right)
\Bigg]^{1/2}
\nonumber\\
&\quad\times
\Bigg[
\Phi_{j,\ell}(0;\theta)
+
2\sum_{h=1}^\infty \Phi_{j,\ell}(h;\theta)
\Bigg].
\end{align}
We now normalize the long-run covariance matrix to obtain the corresponding autocorrelation matrix. For \(j,\ell\in\{1,\ldots,m_1\}\), by using \eqref{eq:explicit_Sigma}, the prefactors $\frac{\sigma_\varepsilon^2}{2\pi(1-\theta^2)}\Big[ \cdots \Big]^{1/2}$ cancel, and therefore
\begin{align}
\label{eq:RGAP_proof}
\mathcal{R}^{\mathrm{GAP}}_{j,\ell}(\theta)
=
\frac{\Phi_{j,\ell}(0;\theta)+2\sum_{h=1}^\infty \Phi_{j,\ell}(h;\theta)}
{\sqrt{\Phi_{j,j}(0;\theta)+2\sum_{h=1}^\infty \Phi_{j,j}(h;\theta)}
 \sqrt{\Phi_{\ell,\ell}(0;\theta)+2\sum_{h=1}^\infty \Phi_{\ell,\ell}(h;\theta)}}.
\end{align}
This proves the first claim.

Assume now that \(k_1=1\). Then, for each channel \(j=1,\ldots, m_1\),
\[
(\mathbf W\mathbf X_t)_j=W^{(1)}_{0,j}X_t,
\]
so that, for every \(h\geq 0\), using \eqref{eq:cross_relu_cov} with $k_1=1$ yields
\[
\rho_{j,\ell}(h,\theta)
=
\mathrm{Corr}\bigl(W^{(1)}_{0,j}X_1,W^{(1)}_{0,\ell}X_{1+h}\bigr)
=
\operatorname{sgn}\!\bigl(W^{(1)}_{0,j}W^{(1)}_{0,\ell}\bigr)\theta^h.
\]

\paragraph{Case $\theta=0$.}
We first evaluate \(\mathcal R_{j,\ell}^{\mathrm{GAP}}(0)\). If
\(W^{(1)}_{0,j}W^{(1)}_{0,\ell}>0\), then \(\rho_{j,\ell}(h,\theta)=\theta^h\), and at \(\theta=0\),
\[
\Phi_{j,\ell}(0;0)=\Phi_{j,j}(0;0)=\Phi_{\ell,\ell}(0;0)=\pi-1,
\qquad
\Phi_{j,\ell}(h;0)=0,\quad h\ge1.
\]
Hence
\[
\mathcal R_{j,\ell}^{\mathrm{GAP}}(0)=1.
\]
If instead \(W^{(1)}_{0,j}W^{(1)}_{0,\ell}<0\), then \(\rho_{j,\ell}(0)=-1\), so
\[
\Phi_{j,\ell}(0;0)
=
\sqrt{1-(-1)^2}+(-1)\bigl(\pi-\arccos(-1)\bigr)-1
=
-1,
\]
while still \(\Phi_{j,j}(0;0)=\Phi_{\ell,\ell}(0;0)=\pi-1\). Therefore
\[
\mathcal R_{j,\ell}^{\mathrm{GAP}}(0)=-\frac{1}{\pi-1}\;.
\]
\paragraph{Case $\theta \uparrow 1$.} If
\(W^{(1)}_{0,j}W^{(1)}_{0,\ell}>0\), then \(\rho_{j,\ell}(h,\theta)=\theta^h\) does not depend on $(j,\ell)$. This means $\Phi_{j,\ell}(h;\theta)=\Phi( \theta^h)$, and both numerator and denominator in \eqref{eq:RGAP_proof} are identical. Hence, for all $\theta \in (-1,1)$, $\mathcal R_{j,\ell}^{\mathrm{GAP}}(\theta)=1\;.$  Whereas, when $W^{(1)}_{0,j}W^{(1)}_{0,\ell}<0$ we have \(\Phi_{j,\ell}(h;\theta)=\Phi(-\theta^h)\), and $\mathcal R_{j,\ell}^{\mathrm{GAP}}(\theta)$ becomes
\begin{equation}
\mathcal R_{j,\ell}^{\mathrm{GAP}}(\theta)
=
\frac{\Phi(-1)+2\sum_{h=1}^\infty \Phi(-\theta^h)}
{\Phi(1)+2\sum_{h=1}^\infty \Phi(\theta^h)}.
\end{equation}
First, we multiply both numerator and denominator by $1-\theta.$ By definition $\theta \in (-1,1)$, so  $1-\theta \neq 0$. Further,  since we are interested in the limit $\theta\uparrow 1$, it is not restrictive to assume $\theta\in (0,1)$. Define, 
$$\qquad \tilde{\Phi}(u)= \begin{cases}
    \frac{\Phi(u)}{u} \quad \text{for }u\in (0,1]\\
    \frac{\pi}{2} \quad \text{for }u=0\;,\\
\end{cases}$$
that is, the continuous extension of $\Phi$ in $[0,1]$. Since $\theta>0$
\begin{align*}
    (1-\theta)\sum_{h=1}^\infty \Phi(\theta^h) = \sum_{h=1}^\infty  (1-\theta) \theta^h \frac{\Phi(\theta^h)}{\theta^h}= \sum_{h=1}^\infty  (1-\theta) \theta^h \tilde \Phi(\theta^h)\;.
\end{align*}
Lemma \ref{lem:Riemann_sums} identifies this as a Riemann sum and yields for $\theta \uparrow 1,$
$$\sum_{h=1}^\infty  (1-\theta) \theta^h \tilde \Phi(\theta^h)\to \int_0^1 \tilde\Phi (u)\, du\;.$$
Putting everything together, 
\begin{align*}
    \lim_{\theta \uparrow 1} \mathcal R_{j,\ell}^{\mathrm{GAP}}(\theta)
=\frac{\int_0^1 \tilde\Phi (-u)\, du}{\int_0^1 \tilde\Phi (u)\, du}\;.
\end{align*}
It remains to compute the two integrals above. 
\begin{align*}
    I_-:&=\int_0^1 \frac{\Phi(-u)}{u}\, du= \int_0^1 \frac{\sqrt{1-u^2}-u(\pi-\arccos{(-u)})-1}{u}\, du\;,\\
    I_+&= \int_0^1 \frac{\Phi(u)}{u}\, du= \int_0^1 \frac{\sqrt{1-u^2}+u(\pi-\arccos{(u)})-1}{u}\, du\;.
\end{align*}
Using $\arccos{(-u)}=\pi-\arccos{u}$, we write
\begin{align*}
    I_-&=\int_0^1 \frac{\sqrt{1-u^2}-1}{u}\,du-\int_0^1 \arccos{(u)}\, du\;\\
    I_+ &= \int_0^1 \frac{\sqrt{1-u^2}-1}{u}\,du+\pi-\int_0^1 \arccos{(u)}\, du\;.
\end{align*}
For the elementary integral $\int_0^1 \arccos{(u)}\, du=1\;.$  Whereas, 
\begin{align*}
    \int_0^1 \frac{\sqrt{1-u^2}-1}{u}\,du= -\int_0^1 \frac{u}{\sqrt{u^2-1}+1}\,du
\end{align*}
Now we use the change of variables $x=\sqrt{1-u^2}$, so that $dx=-\frac{u}{\sqrt{1-u^2}}\,du$, hence
\begin{align*}
     \int_0^1 \frac{-u}{\sqrt{u^2-1}+1}\,du=\int_1^0 \frac{x}{x+1}\,dx=-\int_0^1 \left(1-\frac{1}{1+u}\right)\,du
=
\ln(2)-1.
\end{align*}

It follows
\[
\lim_{\theta\uparrow 1}\mathcal R^{\mathrm{GAP}}_{j,\ell}(\theta)
=
\frac{\log2-2}{\pi+\log2-2}.
\]
\end{proof}

\begin{lemma}
\label{lem:Riemann_sums}
Let \(\Phi\) be a continuous function in $[0,1]$. Then, as $\theta\uparrow 1$,
\[
(1-\theta)\sum_{h=1}^{\infty}\theta^{h}\Phi(\theta^{h})
\longrightarrow
\int_{0}^{1} \Phi(u)\,du.
\]
\end{lemma}

\begin{proof}
For \(h\ge 1\), define
\[
I_h(\theta):=[\theta^h,\theta^{h-1}].
\]
Then the intervals \(I_h(\theta)\) partition \((0,1]\), and
\[
|I_h(\theta)|
:=
\theta^{h-1}-\theta^h
=
(1-\theta)\theta^{h-1}.
\]
Hence $
(1-\theta)\theta^h=\theta\,|I_h(\theta)|, $
so that
\[
(1-\theta)\sum_{h=1}^{\infty}\theta^{h}\Phi(\theta^{h})
=
\theta\sum_{h=1}^{\infty}|I_h(\theta)|\,\Phi(\theta^{h}).
\]
It therefore suffices to prove that
\[
\sum_{h=1}^{\infty}|I_h(\theta)|\,\Phi(\theta^{h})
\longrightarrow
\int_0^1 \Phi(u)\,du,
\qquad \theta\uparrow 1.
\]
Since $\Phi$ is continuous on a compact interval, $\Phi$ is uniform continuous. Fix $\varepsilon>0$. There exist $\delta$ such that $|\Phi(x)-\Phi(y)|\leq \varepsilon$ for all $|x-y|<\delta\;. $ Furthermore, there exists a $\theta_0\in [0,1)$ such that, for all $\theta\in [\theta_0,1),$
$$|I_h(\theta)|=\theta^{h-1}(1-\theta)\leq 1\cdot (1-\theta)\leq \delta\;.$$
Therefore,
\begin{align*}
  \Big| \sum_{h=1}^{\infty}|I_h(\theta)|\,\Phi(\theta^{h})
-
\int_0^1 \Phi(u)\,du\Big|=&  \Big| \sum_{h=1}^\infty\int_{I_h(\theta)} \Phi(\theta^h)  \,du - \sum_{h=1}^\infty \int_{I_h(\theta)} \Phi(u)  \,du\Big|\\
\leq&\sum_{h=1}^\infty \int_{I_h(\theta)} |\Phi(\theta^h)-\Phi(u)|\, du\\
\leq & \varepsilon \sum_{h=1}^\infty \int_{I_h(\theta)} \, du= \varepsilon\;.
\end{align*}
In the last inequality, we used the fact that $u\in I_h(\theta)$, that implies $|\theta^h-u|\leq \delta\;.$ Since $\varepsilon>0$ is arbitrary, this concludes the proof. 
\end{proof}

\begin{lemma}
\label{lem:cov_relu}
    Let \(\begin{pmatrix}
        U \\
        V
    \end{pmatrix} \sim \mathcal{N}\left(\begin{pmatrix}
        0 \\
        0
    \end{pmatrix} , \begin{pmatrix}
        1 & \rho \\
        \rho & 1\\
    \end{pmatrix}\right)\) with correlation \(\rho\in[-1,1]\). Then,
    \[
\Cov(\mathrm{ReLU}(U),\mathrm{ReLU}(V))
=
\frac{1}{2\pi}
\left(
\sqrt{1-\rho^2}
+
\rho\bigl(\pi-\arccos(\rho)\bigr)
-1
\right).
\]
\end{lemma}
\begin{proof}

The Gaussian random vector $(U,V)^\top$ has the same distribution as $$(U',V')^\top:=\left(U, \rho U +\sqrt{1-\rho^2} Z\right)^\top$$ for $Z\sim \mathcal{N}(0,1)$ independent of $U$, that is, $\mathbf{Z}:=(U, Z)^\top \sim \mathcal{N}\left(\mathbf{0}, \mathbf{I}_2\right)\;.$

 Hence, it is enough to derive an explicit expression for
\[
\mathbb E[\mathrm{ReLU}(U')\mathrm{ReLU}(V')].
\]
Let \(\Theta:\mathbb{R}\to\{0,1\}\) denote the Heaviside function $
\Theta(x):=\mathbf{1}_{\{x>0\}}.$
Since \(\mathrm{ReLU}(x)=x\Theta(x)\), Equation~(1) in \citet{NIPS2009_5751ec3e} with \(n=1\),  $\mathbf{x}:=(1,0)^\top$ and $\mathbf{y}:= (\rho, \sqrt{1-\rho^2})$ gives
\[
k_1(\mathbf{x},\mathbf{y})
=
2\int_{\mathbb{R}^2}\frac{1}{2\pi}e^{-\|\mathbf{z}\|^2/2}
\,\Theta(\mathbf{x}^\top\mathbf{z})\Theta(\mathbf{y}^\top\mathbf{z})
(\mathbf{x}^\top\mathbf{z})(\mathbf{y}^\top\mathbf{z})\,d\mathbf{z}.
\]
Since \(\mathbf{Z}\sim \mathcal{N}(\mathbf0,\mathbf{I}_2)\), this can be written as
\[
k_1(\mathbf{x},\mathbf{y})
=
2\,\mathbb{E}\Big[
\Theta(\mathbf{x}^\top\mathbf{Z})\Theta(\mathbf{y}^\top\mathbf{Z})
(\mathbf{x}^\top\mathbf{Z})(\mathbf{y}^\top\mathbf{Z})
\Big].
\]
Using \(U'=\mathbf{x}^\top\mathbf{Z}\) and \(V'=\mathbf{y}^\top\mathbf{Z}\), we obtain $
k_1(\mathbf{x},\mathbf{y})
=
2\,\mathbb{E}\big[\mathrm{ReLU}(U')\mathrm{ReLU}(V')\big].$
Therefore,
\[
\mathbb{E}\big[\mathrm{ReLU}(U)\mathrm{ReLU}(V)\big]
=
\mathbb{E}\big[\mathrm{ReLU}(U')\mathrm{ReLU}(V')\big]
=
\frac12\,k_1(\mathbf{x},\mathbf{y}).
\]
Moreover, Equations~(3) and~(6) in \citet{NIPS2009_5751ec3e} show that
\[
k_1(\mathbf{x},\mathbf{y})
=
\frac{1}{\pi}\|\mathbf{x}\|\,\|\mathbf{y}\|\,J_1(\theta),
\qquad
J_1(\theta)=\sin\theta+(\pi-\theta)\cos\theta,
\]
where \(\theta\) is the angle between \(\mathbf{x}\) and \(\mathbf{y}\). Since \(\|\mathbf{x}\|=\|\mathbf{y}\|=1\) and \(\mathbf{x}^\top\mathbf{y}=\rho\), it follows that
\[
\cos\theta=\rho,
\qquad
\theta=\arccos(\rho),
\qquad
\sin\theta=\sqrt{1-\rho^2}.
\]
Hence
\[
k_1(\mathbf{x},\mathbf{y})
=
\frac{1}{\pi}
\left(
\sqrt{1-\rho^2}
+
\rho\bigl(\pi-\arccos(\rho)\bigr)
\right),
\]
and therefore
\[
\mathbb{E}\big[\mathrm{ReLU}(U)\mathrm{ReLU}(V)\big]
=
\frac{1}{2\pi}
\left(
\sqrt{1-\rho^2}
+
\rho\bigl(\pi-\arccos(\rho)\bigr)
\right).
\]
Finally, since \(U, V\sim \mathcal{N}(0,1)\),
\[
\mathbb E[\mathrm{ReLU}(V)]=\mathbb E[\mathrm{ReLU}(U)]
=
\int_0^\infty u\,\frac{1}{\sqrt{2\pi}}e^{-u^2/2}\,du
=
\frac{1}{\sqrt{2\pi}},
\]
and we obtain
\begin{align*}
\Cov(\mathrm{ReLU}(U),\mathrm{ReLU}(V))
&=
\mathbb{E}\big[\mathrm{ReLU}(U)\mathrm{ReLU}(V)\big]
-
\mathbb E[\mathrm{ReLU}(U)]\mathbb E[\mathrm{ReLU}(V)]\\
&=
\frac{1}{2\pi}
\left(
\sqrt{1-\rho^2}
+
\rho\bigl(\pi-\arccos(\rho)\bigr)
\right)
-\frac{1}{2\pi}.
\end{align*}
\end{proof}

\end{document}